\def\UseBibLatex{1}

\documentclass[12pt]{article}
\newcommand{\notACMVer}[1]{#1}
\newcommand{\ACMVer}[1]{}

\makeatletter
\def\input@path{{styles/}}
\makeatother

\usepackage{myntabbing}%

\usepackage{lmodern}
\usepackage[T1]{fontenc}
\usepackage{bm}%
\usepackage{mleftright}%
\usepackage{xcolor}
\usepackage[bibencoding=utf8,style=alphabetic,backend=biber]{biblatex}%
\usepackage{sariel_biblatex}%

\usepackage{graphicx}
\usepackage{euscript}%
\usepackage{xspace}%
\usepackage{color}%
\usepackage{caption}%
\usepackage{amssymb}%
\usepackage{amsmath}%
\usepackage{paralist}%
\usepackage{enumerate}
\usepackage{multirow}
\usepackage[normalem]{ulem}%

\usepackage{courier}

\notACMVer{%
   \usepackage[hyperref, amsmath,thmmarks]{ntheorem}

   \theoremseparator{.}%
}

\notACMVer{%
   \usepackage{titlesec}
   \titlelabel{\thetitle. }
}

   \numberwithin{figure}{section}
   \numberwithin{table}{section}
   \numberwithin{equation}{section}%

\definecolor{sarielChangeColor}{rgb}{0,0.351,0.055}%

    \newcommand{\translation}{drift\xspace}

\newcommand{\WSPD}{\Term{WSPD}\index{well-separated pair
      decomposition}%
   \xspace}

\newcommand{\True}{\Term{TRUE}\xspace}

\newcommand{\ddSetY}[1]{\DistSetChar\pth{#1}}
\newcommand{\ddSetX}[2]{\DistSetChar\pth{#1, #2}\xspace}

\newcommand{\ddiSetX}[3]{\DistSetChar_{#1}\pth{ #2, #3 }}

\newcommand{\ddiX}[3]{\DistChar_{#1}\pth{#2, #3}}

\newcommand{\ddRankX}[3]{\DistChar^{#1}\pth{#2, #3}}
\newcommand{\ddRankY}[2]{\DistChar^{#1}\pth{#2}}

\newcommand{\ddiRankX}[4]{\DistChar_{#1}^{#2} \pth{#3, #4}}
\newcommand{\ddiRankY}[3]{\DistChar_{#1}^{#2} \pth{#3}}

\newcommand{\Context}{\Gamma}
\newcommand{\target}{\ensuremath{f}}

\newcommand{\kCenter} {\PStyle{{k{}Center}}\xspace}
\newcommand{\kCenterFunc} {\target_{cen}}

\newcommand{\WellBehavedDistanceProblem} {\PStyle{{Nice{}Distance{}Problem}}\xspace}

\newcommand{\NDP}{\PStyle{{NDP}}\xspace}

\newcommand{\FPTAS}{\Term{FP{T}AS}\xspace}

\renewcommand{\th}{\si{th}\xspace}

\providecommand{\si}[1]{#1}

\newcommand{\BenAddress}{%
   Department of Computer Science; University of Illinois; %
   201 N. Goodwin Avenue; %
   Urbana, IL, 61801, USA; %
   {\tt \si{raichel}2\atgen{}\si{uiuc}.\si{edu}}; %
   {\tt \url{http://illinois.edu/\string~raichel2}.%
   }
}%

\newcommand{\BenThanks}[1]{\thanks{%
      \BenAddress{} %
      #1}}

\newcommand{\atgen}{\symbol{'100}}
\newcommand{\SarielAddress}{%
   Department of
   Computer Science; University of Illinois; 201 N. Goodwin Avenue;
   Urbana, IL, 61801, USA; {\tt
      \si{sariel}\atgen{}\si{uiuc.edu}}; %
   {\tt \url{http://illinois.edu/\string~sariel}.
   }%
}

\newcommand{\SarielThanks}[1]{\thanks{%
      \SarielAddress{} %
      #1}}

\definecolor{blue25}{rgb}{0,0,0.55}%
\newcommand{\emphic}[2]{%
   \textcolor{blue25}{%
      \textbf{\emph{#1}}}%
   \index{#2}}

\newcommand{\emphi}[1]{\emphic{#1}{#1}}
\newcommand{\pbrc}[2][\!\!]{#1\left[ {#2} \Bigr. \right]}

\definecolor{red25}{rgb}{0.4,0,0.0}%

\newcommand{\PStyle}[1]{\textcolor{red25}{\textrm{\textsf{#1}}}}

\notACMVer{%
   \numberwithin{figure}{section}

   \newtheorem{theorem}{Theorem}[subsection]
   \newtheorem{lemma}[theorem]{Lemma}%
   \newtheorem{corollary}[theorem]{Corollary}

   {\theorembodyfont{\rm} }
   {\theorembodyfont{\rm} \newtheorem{remark}[theorem]{Remark}}
   {\theorembodyfont{\rm} \newtheorem{definition}[theorem]{Definition}}
   {\theorembodyfont{\rm} \newtheorem{example}[theorem]{Example}}
}
\newtheorem{algorithm}[theorem]{Algorithm}
\newtheorem{problem}[theorem]{Problem}%
\newtheorem{claim}[theorem]{Claim}%

\newcommand{\brc}[1]{\left\{ {#1} \right\}}

\newcommand{\SSet}{\mathsf{S}}
\newcommand{\XSet}{\mathsf{X}}
\newcommand{\ZSet}{\mathsf{Z}}
\newcommand{\RSample}{\mathsf{U}}%

\newcommand{\constNet}{{c_{\mathrm{net}}}}%

\newcommand{\constC}{\eta}%

\newcommand{\cC}{c_3}%
\newcommand{\cD}{c_4}%
\newcommand{\cE}{c_5}%
\newcommand{\cF}{c_6}%

\newcommand{\diameterX}[1]{\mathrm{diam}\pth{#1}}

\newcommand{\DistChar}{\mathtt{d}}
\newcommand{\ddX}[2]{\DistChar\pth{#1, #2}}
\newcommand{\DistSetChar}{\mathtt{D}}

\usepackage{mathcalb}

\newcommand{\dC}{\mathcalb{d}}

\newcommand{\dDirY}[2]{\dC\pth{#1 \rightarrow #2}}
\newcommand{\dmY}[2]{\dC\pth{#1, #2}}
\newcommand{\DistHDY}[2]{\dC_H\pth{#1, #2}}
\newcommand{\dHDY}[2]{\DistHDY{#1}{#2}}

\newcommand{\distChar}{\mathsf{d}}%
\newcommand{\distX}[2]{\distChar\pth{#1, #2}}

\newcommand{\abs}[1]{\left| {#1}  \right|}

\newcommand{\dY}[2]{\left\| #1 - #2 \right\|}

\newcommand{\nRad}{\ell}%
\newcommand{\nxRad}{\mathsf{l}}%

\newcommand{\nnX}[2]{\nu\pth{#1, #2}}%

\newcommand{\ddmX}[3]{\mathsf{d}_{#1}\pth{#2, #3}}

\newcommand{\Term}[1]{\textsf{#1}}%

\newcommand{\LP}{\Term{LP}\xspace}%
\newcommand{\PSPACE}{\Term{PSPACE}\xspace}%
\newcommand{\RAM}{\Term{RAM}\xspace}%

\renewcommand{\P}{\Term{P}\xspace}%

\renewcommand{\Re}{{\rm I\!\hspace{-0.025em} R}}

\newcommand{\HLinkShort}[2]{\hyperref[#2]{#1\ref*{#2}}}
\newcommand{\HLink}[2]{\hyperref[#2]{#1~\ref*{#2}}}
\newcommand{\HLinkPage}[2]{\hyperref[#2]{#1~\ref*{#2}%
      $_\text{p\pageref{#2}}$}}
\newcommand{\HLinkPageOnly}[1]{\hyperref[#1]{Page~\refpage*{#1}%
      $_\text{p\pageref{#1}}$}}

\newcommand{\HLinkSuffix}[3]{\hyperref[#2]{#1\ref*{#2}{#3}}}
\newcommand{\HLinkPageSuffix}[3]{\hyperref[#2]{#1\ref*{#2}%
      #3$_\text{p\pageref{#2}}$}}

\newcommand{\figlab}[1]{\label{fig:#1}}
\newcommand{\figref}[1]{\HLink{Figure}{fig:#1}}

\newcommand{\seclab}[1]{\label{sec:#1}}
\newcommand{\secref}[1]{\HLink{Section}{sec:#1}}

\newcommand{\corlab}[1]{\label{cor:#1}}
\newcommand{\corref}[1]{\HLink{Corollary}{cor:#1}}%

\providecommand{\deflab}[1]{\label{def:#1}}
\newcommand{\defref}[1]{\HLink{Definition}{def:#1}}
\newcommand{\defrefY}[2]{\hyperref[def:#1]{#2}}

\newcommand{\clmlab}[1]{\label{claim:#1}}
\newcommand{\clmref}[1]{\HLink{Claim}{claim:#1}}

\newcommand{\remlab}[1]{\label{rem:#1}}
\newcommand{\remref}[1]{\HLink{Remark}{rem:#1}}%

\newcommand{\lemlab}[1]{\label{lemma:#1}}
\newcommand{\lemref}[1]{\HLink{Lemma}{lemma:#1}}%

\newcommand{\linelab}[1]{\label{line:#1}}%
\newcommand{\lineref}[1]{\HLink{Line}{line:#1}}%

\newcommand{\tablab}[1]{\label{table:#1}}%
\newcommand{\tabref}[1]{\HLink{Table}{table:#1}}%

\newcommand{\alglab}[1]{\label{Algorithm:#1}}%
\newcommand{\algref}[1]{\HLink{Algorithm}{Algorithm:#1}}%

\newcommand{\thmlab}[1]{{\label{theo:#1}}}
\newcommand{\thmref}[1]{\HLink{Theorem}{theo:#1}}

\providecommand{\eqlab}[1]{}%
\renewcommand{\eqlab}[1]{\label{equation:#1}}

\notACMVer{%
\newenvironment{proof}{\trivlist\item[]\emph{Proof}:}%
   {\unskip\nobreak\hskip 1em plus 1fil\nobreak%
      \myqedsymbol%
      \parfillskip=0pt%
      \endtrivlist}
}

{\unskip\nobreak\hskip 1em plus 1fil\nobreak%
   \myqedsymbol%
   \parfillskip=0pt%
   \endtrivlist}

\newcommand{\myqedsymbol}{\rule{2mm}{2mm}}
\newcommand{\cardin}[1]{\left| {#1} \right|}
\newcommand{\pbrcx}[1]{\left[ {#1} \right]}
\newcommand{\Ex}[2][\!]{\mathop{\mathbf{E}}#1\pbrcx{#2}}
\newcommand{\Prob}[1]{\mathop{\mathbf{Pr}}\!\pbrcx{#1}}
\newcommand{\sep}[1]{\left|\, {#1} \bigr. \right.}

\newcommand{\Set}[2]{\left\{ #1 \;\middle\vert\; #2 \right\}}

\providecommand{\ds}{\displaystyle}

\newcommand{\eps}{\varepsilon}%

\newcommand{\Family}{\mathcal{F}}

\providecommand{\remove}[1]{}

\newcommand{\ceil}[1]{\left\lceil {#1} \right\rceil}
\newcommand{\floor}[1]{\left\lfloor {#1} \right\rfloor}

\newcommand{\etal}{\textit{et~al.}\xspace}

\usepackage{pifont}%

\usepackage[symbol*]{footmisc}

\usepackage{hyperref}%
\usepackage[ocgcolorlinks]{ocgx2}

\hypersetup{%
      unicode,
      breaklinks,%
      colorlinks=true,%
      urlcolor=[rgb]{0.25,0.0,0.0},%
      linkcolor=[rgb]{0.5,0.0,0.0},%
      citecolor=[rgb]{0,0.2,0.445},%
      filecolor=[rgb]{0,0,0.4},
      anchorcolor=[rgb]={0.0,0.1,0.2}%
}
\usepackage[all]{hypcap}

\DefineFNsymbols*{stars}[text]{%
   {\ding{172}} %
   {\ding{173}} %
   {\ding{174}} %
   {\ding{175}} %
   {\ding{176}} %
   {\ding{177}} %
   {\ding{178}} %
   {\ding{179}} %
   {\ding{180}} %
   }

\usepackage{microtype}
\usepackage{xurl}

\newcommand{\pth}[2][\!]{\mleft({#2}\mright)}

\newcommand{\GridNbrChar}{\mathsf{N}}
\newcommand{\GridNbr}[2]{\mathsf{N}_{\leq #2}\pth{#1}}

\newcommand{\gidX}[1]{\mathrm{id}\pth{#1}}

\newcommand{\pnt}{{\mathsf{p}}}%
\newcommand{\pntA}{{\mathsf{q}}}%
\newcommand{\pntB}{{\mathsf{s}}}

\newcommand{\pntX}{{\mathsf{x}}}

\newcommand{\PntSet}{\mathsf{P}}%

\newcommand{\PS}{\mathsf{P}}%
\newcommand{\R}{\mathsf{R}}%
\newcommand{\B}{\mathsf{B}}%

\newcommand{\PntSetA}{\mathsf{W}}%
\newcommand{\PntSetB}{\mathsf{Q}}%

\renewcommand{\th}{th\xspace}

\newcommand{\centers}{C}%

\newcommand{\cset}{\mathcal{C}}
\newcommand{\ropt}{r_{opt}}

\newcommand{\ndpAlg}{\AlgorithmRef{{ndpAlg}}{ndp_Alg}\xspace}

\newcommand{\delete}{\AlgorithmRef{delFar}{delete}\xspace}

\newcommand{\net}{\AlgorithmRef{net}{net}\xspace}

\DefineNamedColor{named}{RedViolet} {cmyk}{0.07,0.90,0,0.34}
\providecommand{\AlgorithmI}[1]{{\textcolor[named]{RedViolet}{\texttt{\bf{#1}}}}}
\providecommand{\Algorithm}[1]{{\AlgorithmI{#1}\index{algorithm!#1@{\AlgorithmI{#1}}}}}
\providecommand{\AlgorithmRef}[2]{%
   \hyperref[algorithm:#2]%
   {\AlgorithmI{#1}%
      \index{algorithm!#1@{\AlgorithmI{#1}}}%
   }%
}

\newcommand{\AlgorithmAnchor}[1]{%
   \phantomsection
   \label{algorithm:#1}}

\newcommand{\polylog}{\operatorname{polylog}}

\newcommand{\netX}[2]{\net\pth{#1, #2}}

\newcommand{\gset}{\mathcal{G}}
\newcommand{\Grid}{\mathsf{G}}

\newcommand{\MST}{\ensuremath{\mathsf{MST}}\xspace}

\newcommand{\Graph}{{G}}%

\newcommand{\FarX}[2]{{#1}^{\geq #2}}
\newcommand{\CloseX}[2]{#1^{< #2}}
\newcommand{\CloseY}[2]{#1^{\leq #2}}

\newcommand{\oracle}[1]{\Algorithm{orac}}%
\newcommand{\descrip}[1]{\mathrm{sk}\pth{#1}}%

\newcommand{\Interval}{\mathcal{I}}
\newcommand{\kthDist} {\PStyle{{$k$\th{}Distance}}\xspace}

\newcommand{\distinct} {\PStyle{{Smallest{}\si{NonZeroDist}}}\xspace}

\newcommand{\ccx}{\mathcal{C}}
\newcommand{\CCX}[2]{\ccx_{\leq #2}\pth{#1}}

\newcommand{\Partition}{\mathcal{P}}
\newcommand{\PartitionA}{\mathcal{Q}}
\newcommand{\PartitionOpt}{\mathcal{P}_\mathrm{opt}}

\newcommand{\fLEMST}{\ell_{MST}}

\newcommand{\ball}{\mathsf{b}}

\newcommand{\ballCR}[2]{\mathsf{b}\pth{#1,#2}}
\newcommand{\vecA}{\psi}

\newcommand{\decider}{\Algorithm{{decider}}\xspace}

\newcommand{\contextUpdate}{\Algorithm{context{}Update}\xspace}

\newcommand{\WeightX}[2][\!]{\omega\pth[#1]{#2}}

\newcommand{\NWeightX}[1]{{\omega}_{\GridNbrChar}\pth{#1}}
\newcommand{\cell}{\Box}

\newcommand{\rmin}{\mathsf{r}_\mathrm{min}}%
\newcommand{\NetSet}{\mathcal{N}}%
\newcommand{\Schonhage}{Sch{\"o}nhage\xspace}

\newlength{\savedparindent}
\newcommand{\SaveIndent}{\setlength{\savedparindent}{\parindent}}
\newcommand{\RestoreIndent}{\setlength{\parindent}{\savedparindent}}

\newcommand{\Frechet}{Fr\'{e}chet\xspace}%

\providecommand{\CodeComment}[1]{\textcolor{blue}{\texttt{#1}}}

\newcommand{\constNetVal}{37}
\newcommand{\constCVal}{28}

\newcommand{\cDecider}{\varphi}

\newcommand{\SpreadX}[1]{\mathrm{sprd}\pth{#1}}

\newcommand{\estLogDist}{\AlgorithmRef{{est{}Log{}Dist}}{est_log_dist}\xspace}

\newcommand{\medianNN}{\AlgorithmRef{midNN}{midNN}\xspace}
\newcommand{\medianNNExact}{\Algorithm{midNNExact}\xspace}
\newcommand{\smallComp}{\AlgorithmRef{{smallComp}}{small_comp}\xspace}
\newcommand{\lowSpread}{\AlgorithmRef{{lowSpread}}{log_n_n}\xspace}

\newcommand{\DTbl}[1]{%
   \begin{minipage}{0.55\linewidth}
       \smallskip%
       #1%
       \smallskip
   \end{minipage}
}

\newcommand{\PDecider}{\textsf{Decider}}
\newcommand{\PLipschitz}{\textsf{Lipschitz}}
\newcommand{\PPrune}{\textsf{Prune}}

\newcommand{\PDeciderRef}{\PDecider\xspace}
\newcommand{\PLipschitzRef}{\PLipschitz\xspace}
\newcommand{\PPruneRef}{\PPrune\xspace}

\newcommand{\itemPDecider}{\item[\quad\PDecider:]}
\newcommand{\itemPLipschitz}{\item[\quad\PLipschitz:]}
\newcommand{\itemPPrune}{\item[\quad\PPrune:]}

\makeatletter
\def\xnamedlabel#1#2{\begingroup
   \phantomsection%
   \def\@currentlabel{#2}%
   \label{#1}\endgroup
}
\makeatother

\newcommand{\clinelab}[2]{%
}%

\newcommand{\Do}{{\small\bf do}\ }

\newcommand{\Return}{{\small\bf return\ }}

\newcommand{\If}{{\small\bf if}\ }
\newcommand{\Then}{{\small\bf then}\ }

\newcommand{\While}{{\small\bf while}\ }

\newcommand{\res}{\mathrm{res}}
\newcommand{\reslog}{\res_{\log}}

\newcommand{\xbeginlgox}{\begin{minipage}{1in}\begin{tabbing}
           \quad\=\qquad\=\qquad\=\qquad\=\qquad\=\qquad\=\qquad\=\kill}
        \newcommand{\xendlgox}{\end{tabbing}\end{minipage}}

\newenvironment{myprogram}{
   \begin{minipage}{4.0in}
   \begin{tabbing}
       \ \ \ \ \= \ \ \ \= \ \ \ \ \= \ \ \ \ \= \ \ \ \ \=
      \ \ \ \ \= \ \ \ \ \= \ \ \ \ \= \ \ \ \ \=
      \ \ \ \ \= \ \ \ \ \= \ \ \ \ \= \ \ \ \ \= \kill
}{
   \end{tabbing}
   \end{minipage}
}
\newenvironment{nprogram}{%
   \begin{minipage}{0.9\linewidth}
   \begin{ntabbing}
       \reset%
       \ \ \ \ \= \ \ \ \= \ \ \ \ \= \ \ \ \ \= \ \ \ \ \=
      \ \ \ \ \= \ \ \ \ \= \ \ \ \ \= \ \ \ \ \=
      \ \ \ \ \= \ \ \ \ \= \ \ \ \ \= \ \ \ \ \= \kill
}{
   \end{ntabbing}
   \end{minipage}
}

\providecommand{\TPDF}[2]{\texorpdfstring{#1}{#2}}

\makeatletter
\newlength\tdima
\newcommand\tabright[1]{%
      \setlength\tdima{\linewidth}%
      \addtolength\tdima{\@totalleftmargin}%
      \addtolength\tdima{-\dimen\@curtab}%
      \makebox[\tdima][r]{#1}}
\makeatother

\newcommand{\deciderRank}{\AlgorithmRef{deciderM}{decider_rank}\xspace}

\newcommand{\ItemSpacing}{\quad\quad}

\makeatletter
\def\theorem@checkbold{}
\makeatother

\newcommand{\qedhere}{\ifmmode\qed\else\hfill\qedsymbol\fi}

\usepackage{xcolor}
\usepackage[framemethod=default]{mdframed}
\definecolor{arxivblue}{RGB}{0, 0, 100} %
\surroundwithmdframed[
    linecolor=arxivblue,
    linewidth=0.8pt,
    backgroundcolor=blue!3,   %
    innertopmargin=8pt,
    innerbottommargin=5pt,   %
    innerleftmargin=10pt,
    innerrightmargin=10pt,
    skipabove=2pt,          %
    skipbelow=0pt           %
]{proof}
\AfterEndEnvironment{proof}{\addvspace{9pt}\noindent}

\bibliography{net_prune}

\begin{document}

\title{%
   Net and Prune%
   : A Linear Time Algorithm for Euclidean Distance Problems%
      \footnote{Work on this paper was partially supported by NSF AF awards CCF-0915984 and CCF-1217462.
         The full updated version of this paper is available online \cite{hr-nplta-14}. %
         A preliminary version of this paper appeared in STOC 2012 \cite{hr-nplta-13}. %
         There are also journal \cite{hr-nplta-15} and arXiv \cite{hr-nplta-14} versions.
         Note that the arXiv version was recently updated with the results of \secref{Hausdorff}.
      }%
}%

\author{%
   Sariel Har-Peled%
   \ACMVer{\affil{University of Illinois, Urbana-Champaign}}%
   \notACMVer{\SarielThanks{}}%
   \notACMVer{\and}%
   Benjamin Raichel%
   \ACMVer{\affil{University of Illinois, Urbana-Champaign}}%
   \notACMVer{\BenThanks{}}%
}

\date{\today}

\notACMVer{\maketitle}

\begin{abstract}
    We provide a general framework for getting expected linear time constant factor approximations (and in many cases \FPTAS) to several well-known problems in Computational Geometry, such as $k$-center clustering and farthest nearest neighbor. The new approach is robust to variations in the input problem, and yet it is simple, elegant, and practical. In particular, many of these well-studied problems, which fit easily into our framework, either previously had no linear time approximation algorithm, or required rather involved algorithms and analysis. A short list of the problems we consider includes farthest nearest neighbor, $k$-center clustering, smallest disk enclosing $k$ points, Hausdorff distance, $k$\th largest distance, $k$\th smallest $m$-nearest neighbor distance, $k$\th heaviest edge in the \MST, and other spanning forest-type problems, problems involving upward closed set systems, and more.  Finally, we show how to extend our framework such that the linear running time bound holds with high probability.
\end{abstract}

\section{Introduction}

In many optimization problems, one is given a geometric input and the goal is to compute some real-valued function defined over it.  For example, for a finite point set, the function might be the distance between the closest pair of points.  For such a function, there is a search space associated with the function, encoding possible candidate values, and the function returns the minimum value over all possible solutions. For the closest pair problem, the search space is the set of all pairs of input points and their associated distances, and the target function computes the minimum pairwise distance in this space. As such, computing the value of such a function requires searching for the optimal solution, and returning its value; that is, \emph{optimization}.  Other examples of such functions include: \smallskip%
\begin{compactenum}[\quad(A)]
    \item The minimum cost of a clustering of the input. Here, the search space is the possible partitions of the data into clusters together with their associated prices.
    \item The price of a minimum spanning tree. Here, the search space is the set of all possible spanning trees of the input (and their associated prices).
    \item The radius of the smallest enclosing ball. Here, the search space is the set of all possible balls that enclose the input and their radii (i.e., it is not necessarily finite).
\end{compactenum}
\smallskip%
Naturally, many other problems can be described in this way.

It is often possible to construct a decision procedure for such an optimization problem, which, given a query value, can decide whether the query is smaller or larger than the value of the function, without computing the function explicitly.  Naturally, one would like to use this decider to compute the function itself by (say) performing a binary search, to compute the minimum value which is still feasible (which is the value of the function). However, often this is inherently not possible as the set of possible query values is a real interval (which is not finite).  Instead, one first identifies a finite set of critical values that must contain the optimum value, and then searches over these values.%

Naturally, searching over these values directly can be costly, as often the number of such critical values is much larger than the desired running time. An alternative is to attempt to perform an implicit search.

One of the most powerful techniques to solve optimization problems efficiently in Computational Geometry, using such an implicit search, is the \emphi{parametric search} technique of Megiddo \cite{m-apcad-83}.  It is relatively complicated, as it involves implicitly extracting values from a simulation of a parallel decision procedure (often a parallel sorting algorithm).  For this reason, it is inherently not possible for parametric search to lead to algorithms that run faster than $O(n\log n)$ time.  Nevertheless, it is widely used in designing efficient geometric optimization algorithms, see \cite{ast-apsgo-94,s-ps-97}. Luckily, in many cases, one can replace parametric search by simpler techniques (see prune-and-search below, for example), and in particular, it can be replaced by randomization, see the work by van Oostrum and Veltkamp \cite{ov-psmp-04}. Another example of replacing parametric search by randomization is the new simplified algorithm for the \Frechet distance \cite{hr-fdre-11}.  Surprisingly, sometimes these alternative techniques can actually lead to linear-time algorithms.

\SaveIndent

\paragraph{Linear time algorithms.} %
There seem to be three main ways to get linear-time algorithms for geometric optimization problems (exact or approximate):

\begin{compactenum}[\quad(A)]%
    \RestoreIndent

    \item \textbf{Coreset/sketch}. One can quickly extract a compact sketch of the input that contains the desired quantity (either exactly or approximately). As an easy example, consider the problem of computing the axis-parallel bounding box of a set of points - an easy linear scan suffices. There is by now a robust theory of what quantities one can extract a coreset of small size for, such that one can do the (approximate) calculation on the coreset, where usually the coreset size depends only on the desired approximation quality. This approach leads to many linear time algorithms, from shape fitting \cite{ahv-aemp-04}, to $(1+\eps)$-approximate $k$-center/median/mean clustering \cite{h-cm-01-conf,h-cm-04,hm-ckmkm-04,hk-sckmk-07} in constant dimension, and many other problems \cite{ahv-gavc-05}. The running times of the resulting algorithms are usually $O\pth{ n + \mathrm{func}( \text{sketch size})}$. The limitation of this technique is that there are problems for which there is no small sketch, from clustering when the number of clusters is large, to problems where there is no sketch at all \cite{h-ncnc-04} -- for example, for finding the closest pair of points, one needs all the given input, and no sketching is possible.

    \item \textbf{Prune and search}. Here, one prunes away a constant fraction of the input and continues the search recursively on the remaining input. The paramount example of such an algorithm is the linear-time median finding algorithm. However, one can interpret many of the randomized algorithms in Computational Geometry as pruning algorithms \cite{c-arscg-88,cs-arscg-89,m-cgitr-94}.  For example, linear programming in constant dimension in linear time \cite{m-lpltw-84}, and its extension to \LP-type problems \cite{sw-cblpr-92,msw-sblp-96}. Intuitively, \LP-type problems include low-dimensional convex programming (a standard example is the smallest enclosing ball of a point set in constant dimension). However, surprisingly, such problems also include problems that are not convex in nature -- for example, deciding if three points can pierce a set of (axis-parallel) rectangles is an \LP-type problem. Other examples of prune-and-search algorithms that work in linear time include
    \begin{inparaenum}[(i)]
        \item computing an ear in a triangulation of a polygon \cite{eet-seupr-93},
        \item searching in sorted matrices \cite{fj-gsrsm-84}, and
        \item ham-sandwich cuts in two dimensions \cite{lms-ahsc-94}.
    \end{inparaenum}
    Of course, there are many other examples of using prune and search with running time that is super-linear.

    \item \textbf{Grids.} %
    Rabin \cite{r-pa-76} used randomization, the floor function, and hashing to compute the closest pair of a set of points in the plane, in expected linear time.  Golin \etal{} \cite{grss-sracp-95} presented a simplified version of this algorithm, and Smid provides a survey of algorithms on closest pair problems \cite{s-cppcg-00}. A prune and search variant of this algorithm was suggested by Khuller and Matias \cite{km-srsac-95}.  Of course, the usage of grids and hashing to perform approximate point-location is quite common in practice. By itself, this is already sufficient to break lower bounds in the comparison model, for example, for $k$-center clustering \cite{h-cm-04}.  The only direct extension of Rabin's algorithm the authors are aware of is the work by Har-Peled and Mazumdar \cite{hm-facsk-05} showing a linear time $2$-approximation to the smallest ball containing $k$ points (out of $n$ given points).

    There is some skepticism of algorithms using the floor function, since \Schonhage \cite{s-pram-79} showed how to solve a \PSPACE complete problem, in polynomial time, using the floor function in the real \RAM model -- the trick being packing many numbers into a single word (which can be arbitrarily long in the \RAM model, and still each operation on it takes only constant time). Note that as Rabin's algorithm does not do any such packing of numbers (i.e., its computation model is considerably more restricted), this criticism does not seem to be relevant in the case of this algorithm and its relatives.  %
\end{compactenum}

\medskip

In this paper, we present a new technique that combines all of the above techniques to yield linear-time approximation algorithms.

\paragraph{Nets.} %
Given a point set $\PntSet$, an $r$-net $\NetSet$ of $\PntSet$ is a subset of $\PntSet$ that represents well the structure of $\PntSet$ in resolution $r$.  Formally, we require that for any point in $\PntSet$ there is a net point in distance at most $r$ from it, and no two net points are closer than $r$ to each other, see \secref{compNets} for a formal definition. Thus, nets provide a sketch of the point set for distances that are $r$ or larger. Nets are a useful tool in presenting point sets hierarchically.  In particular, computing nets of different resolutions and linking between different levels leads to a tree-like data structure that can be used to facilitate many tasks, see for example the net-tree \cite{kl-nnsap-04,hm-fcnld-06} for such a data structure for doubling metrics.  Nets can be defined in any metric space, but in Euclidean space, a grid can sometimes provide an equivalent representation.  In particular, net-trees can be interpreted as an extension of (compressed) quadtrees to more abstract settings.

Computing nets is closely related to $k$-center clustering. Specifically, Gonzalez~\cite{g-cmmid-85} shows how to compute an approximate net that has $k$ points in $O(n k)$ time, which is also a $2$-approximation to the $k$-center clustering. This was later improved to $O(n)$ time, for low-dimensional Euclidean space \cite{h-cm-04}, if $k$ is sufficiently small (using grids and hashing).  Har-Peled and Mendel showed how to preprocess a point set in a metric space with constant doubling dimension, in $O( n \log n)$ time, such that an (approximate) $r$-net can be extracted in (roughly) linear time in the size of the net.

\subsection*{Our contribution.}
In this paper, we consider problems of the following form: Given a set $\PntSet$ of weighted points in $\Re^d$, one wishes to solve an optimization problem whose solution is one of the pairwise distances of $\PntSet$ (or ``close'' to one of these values).  Problems of this kind include computing the optimal $k$-center clustering, or the length of the $k$\th edge in the \MST of $\PntSet$, and many others. Specifically, we are interested in problems for which there is a fast approximate decider. That is, given a value $r>0$, we can, in linear time, decide if the desired value is (approximately) smaller than $r$ or larger than $r$.  The goal is then to use this decider to approximate the optimum solution in linear time.  As a first step towards a linear time approximation algorithm for such problems, we point out that one can compute nets in linear time in $\Re^d$, see \secref{compNets}.

However, even if we could implicitly search over the critical values (which we cannot), then we still would require a logarithmic number of calls to the decider, which would yield a running time of $O(n\log n)$, as the number of critical values is at least linear (and usually polynomial).  So instead we use the return values of the decision procedure as we search to thin out the data so that future calls to the decision procedure become successively cheaper.  However, we still cannot search over the critical values (since there are too many of them), and so we also introduce random sampling in order to overcome this.

\paragraph{Outline of the new technique.} %
The new algorithm works by randomly sampling a point and computing the distance to its nearest neighbor.  Let this distance be $r$.  Next, we use the decision procedure to decide if we are in one of the following three cases. %
\medskip %
\begin{compactenum}[\quad(A)]
    \item \textbf{Net.} If $r$ is too small, then we zoom out to a resolution of $r$ by computing an $r$-net and continuing the computation on the net instead of on the original point set. That is, we net the point set into a smaller point set, such that one can solve the original problem (approximately) on this smaller sketch of the input.

    \item \textbf{Prune.} If $r$ is too large, then we remove all points whose nearest neighbor is further than $r$ away (of course, this implies we should only consider problems for which such pruning does not affect the solution).  That is, we isolate the optimal solution by pruning away irrelevant data -- this is similar in nature to what is being done by prune-and-search algorithms.

    \item \textbf{Done.} The value of $r$ is the desired approximation.
\end{compactenum}
\medskip%
We then continue recursively on the remaining data.  In either case, the number of points being handled (in expectation) goes down by a constant factor, and thus the overall expected running time is linear.  In particular, getting this constant factor decrease is the reason we chose to sample from the set of nearest neighbor distances rather than from the set of all inter-point distances.

\paragraph{Significance of results.} %
Our basic framework is presented in a general enough manner to cover, and in many cases greatly simplify, many problems for which linear-time algorithms had already been discovered.  At the same time, the framework provides new linear-time algorithms for an extensive collection of problems, for which previously no linear-time algorithm was known.  The framework should also lead to algorithms for many other problems that are not mentioned.

At a conceptual level, the basic algorithm is simple enough (with its basic building blocks already having efficient implementations) to be highly practical from an implementation standpoint.  Perhaps more importantly, with increasing shifts toward large data sets, algorithms with super-linear running time can be impractical.  Additionally, our framework seems amenable to distributed implementation in frameworks like MapReduce.  Indeed, every iteration of our algorithm breaks the data into grid cells, a step that is similar to the map phase.  In addition, the aggressive thinning of the data by the algorithm guarantees that after the first few iterations, the algorithm is resigned to working on only a tiny fraction of the data.

\paragraph{Significance of the netting step.}
Without the net stage, our framework is already sufficient to solve the closest pair problem, and in this specific case, the algorithm boils down to the one provided by Khuller and Matias \cite{km-srsac-95}. This case seems to be the only application of the new framework where the netting is not necessary. In particular, the broad applicability of the framework seems to be the result of the netting step.

\paragraph{Framework and results.} %
We provide a framework that classifies which optimization problems can be solved using the new algorithm. We get the following new algorithms (all of them have an expected linear running time, for any fixed $\eps$): \smallskip%
\begin{compactenum}[ (A)]
    \item \textbf{$k$-center clustering} (\secref{k_center}). We provide an algorithm that $2$-approximates the optimal $k$-center clustering of a point set in $\Re^d$. Unlike the previous algorithm \cite{h-cm-04} that was restricted to $k= O(n^{1/6})$, the new algorithm works for any value of $k$. This new algorithm is also simpler.

    \item \textbf{$k$\th smallest distance} (\secref{k_th_distance}).
    In the distance selection problem, given a set of points in $\Re^d$, one would like to compute the $k$\th smallest distance defined by a pair of points of $\PntSet$.  It is believed that such exact distance selection requires $\Omega\pth{n^{4/3}}$ time in the worst case \cite{e-rcsgp-95}, even in the plane (in higher dimensions the bound deteriorates). We present an $O(n/\eps^d)$ time algorithm that $(1+\eps)$-approximates the $k$\th smallest distance. Previously, Bespamyatnikh and Segal \cite{bs-faad-02} presented an $O(n \log n + n/\eps^d )$ time algorithm using a well-separated pairs decomposition (see also \cite{dhw-afdrc-12}).

    Given two sets of points $\PntSet$ and $\PntSetA$ with a total of $n$ points, using the same approach, we can $(1+\eps)$-approximate the $k$\th smallest distance in a bichromatic set of distances
    \begin{math}
        X = \Set{ \distX{\pnt}{\pntA}}{\pnt \in \PntSet, \pntA \in \PntSetA}.
    \end{math}

    Intuitively, the mechanism behind distance selection underlines many optimization problems, as it (essentially) performs a binary search over the distances induced by a set of points. As such, being able to do approximate distance selection in linear time should lead to faster approximation algorithms that perform a similar search over such distances. %

    \item \textbf{The $k$\th smallest $m$-nearest neighbor distance} (\secref{k_smallest_m_n_n}).
    For a set $\PntSet = \brc{\pnt_1,\ldots, \pnt_n} \subseteq \Re^d$ of $n$ points, and a point $\pnt \in \PntSet$, its \emphi{$m$\th nearest neighbor} in $\PntSet$ is the $m$\th closest point to $\pnt$ in $\PntSet\setminus \brc{\pnt}$. In particular, let $\ddmX{m}{\pnt}{\PntSet}$ denote this distance.  Here, consider the set of these distances defined for each point of $\PntSet$; that is, $X = \brc{ \ddmX{m}{\pnt_1}{\PntSet}, \ldots, \ddmX{m}{\pnt_n}{\PntSet}}$.  We can approximate the $k$\th smallest number in this set in linear time.

    \item \textbf{Approximate Hausdorff distance}\footnote{This application did not appear in the journal and conference versions of this paper. Although the result was implicitly already there, see \remref{h_is_new}.}
    Given two point sets $R$ (red) and $B$ (blue) of total size $n$, their \emph{Hausdorff distance} is the minimum distance $\ell$, such that all the points of $R$ are at a distance at most $\ell$ from points of $B$, and vice versa. In two dimensions, computing the Hausdorff distance can be done in $O(n\log n)$ time. Still, in higher dimensions, the natural $(1+\eps)$-approximation algorithm uses approximate nearest-neighbor queries. It requires $\Omega( n \log n+ n \tfrac{1}{\eps^d} \log \tfrac{1}{\eps} )$ time (by computing for each point its nearest neighbor in the other point-set). Our framework provides $(1+\eps)$-approximation algorithm for the Hausdorff distance in $\Re^d$, with running time $O(n /\eps^d)$.  See \thmref{h_really_2} and \remref{h_is_new} for details.

    \item \textbf{Exact nearest neighbor distances} (\secref{f_n_n}).
    For the special case where $m=1$, one can turn the above approximation into an exact computation of the $k$\th nearest neighbor distance with only a minor post-processing grid computation.

    As an application, when $k=n$, one can compute in linear time, exactly, the {furthest nearest neighbor distance}; that is, the nearest neighbor distance of the point whose nearest neighbor is furthest away. This measure can be useful, for example, in meshing applications where such a point is a candidate for a region where the local feature size is large and further refinement is needed.

    We are unaware of any previous work directly on this problem, although one can compute this quantity by solving the all-nearest-neighbor problem, which can be done in $O(n \log n)$ time \cite{c-faann-83}.  This algorithm is to some extent the ``antithesis'' to Rabin's algorithm for the closest pair problem, and it is somewhat surprising that it can also be solved in linear time.

    \item \textbf{The $k$\th longest \MST edge} %
    (\secref{longest_MST}).
    Given a set $\PntSet$ of $n$ points in $\Re^d$, we can $(1+\eps)$-approximate, in $O(n/\eps^d)$ time, the $k$\th longest edge in the \MST of $\PntSet$.

    \item \textbf{Smallest ball with a monotone property} (\secref{min_cluster}). %
    Consider a property defined over a set of points, $\PntSet$, that is monotone; that is, if $\PntSetA \subseteq \PntSetB \subseteq \PntSet$ has this property, then $\PntSetB$ must also have this property. Consider such a property that can be easily checked, for example, (i) whether the set contains $k$ points, or (ii) if the points are colored, whether all colors are present in the given point set.  Given a point set $\PntSet$, one can $(1+\eps)$-approximate, in $O(n /\eps^d)$ time, the smallest radius ball $\ball$, such that $\ball \cap \PntSet$ has the desired property. For example, we get a linear time algorithm to approximate the smallest ball enclosing $k$ points of $\PntSet$. The previous algorithm for this problem \cite{hm-facsk-05} was significantly more complicated.  Furthermore, we can approximate the smallest ball such that all colors appear in it (if the points are colored), or the smallest ball such that at least $t$ different colors appear in it, etc.

    More generally, the kind of monotone properties supported are \emph{sketchable}; that is, properties for which there is a small summary of a point set that enables one to decide if the property holds, and given summaries of two disjoint point sets, the summary of the union point set can be computed in constant time.  We believe that formalizing this notion of sketchability is a useful abstraction. See \secref{sketchable} for details.

    \item \textbf{Smallest connected component with a monotone property} %
    (\secref{min_cluster_c}). %
    Consider the connected components of the graph where two points are connected if they are at most $r$ distance from each other. Using our techniques, one can approximate, in linear time, the smallest $r$ such that there is a connected component of this graph for which a required sketchable property holds for the points in this connected component.

    As an application, consider ad hoc wireless networks. Here, we have a set $\PntSet$ of $n$ nodes and their locations (say in the plane), and each node can broadcast in a certain radius $r$ (the larger the $r$, the higher the energy required, so naturally we would like to minimize it). Assume there are two special nodes.  It is natural to ask for the minimum $r$ such that there is a connected component of the above graph that contains both nodes.  That is, these two special nodes can send a message to each other, by message hopping (with distance at most $r$ at each hop). We can approximate this connectivity radius in linear time.

    \item \textbf{Clustering for a monotone property} (\secref{cluster_monotone}). %
    Imagine that we want to break the given point set into clusters, such that the maximum diameter of a cluster is minimized (as in $k$-center clustering). Furthermore, the points assigned to each cluster comply with some sketchable monotone property. We present a $(4+\eps)$-approximation algorithm for these types of problems that runs in $O(n /\eps^d)$ time. This algorithm can be applied to compute lower bounded clustering (i.e., every cluster must contain at least $\alpha$ points), for which the authors recently presented an $O(n \log n)$ approximation algorithm \cite{ehr-fclb-12}. One can get a $2$-approximation using network flow, but the running time is significantly worse \cite{apftk-aac-10}.  See \secref{cluster_monotone_examples} for examples of clustering problems that can be approximated using this algorithm.

    \item \textbf{Connectivity clustering for a monotone property} (\secref{cluster_spanning}). %
    Consider the problem of computing the minimum $r$, such that each connected component (of the graph where points in distance at most $r$ from each other are adjacent) has some sketchable monotone property. We approximate the minimum $r$ for which this holds in linear time.

    An application of this for ad hoc networks is the following -- we have a set $\PntSet$ of $n$ wireless clients, and some of them are base stations; that is, they are connected to the outside world. We want to find the minimum $r$ such that each connected component of this graph contains a base station.

    \item \textbf{Closest pair and smallest non-zero distance} %
    (\secref{closest_non_zero_dist}). %
    Given a set of points in $\Re^d$, consider the problem of finding the smallest non-zero distance defined by these points. This problem is an extension of the closest pair distance, as there might be many identical points in the given point set. We provide a linear-time algorithm for computing this distance \emph{exactly}, which follows easily from our framework.
\end{compactenum}

\paragraph{High probability.} %
Finally, we show in \secref{high_prob} how to modify our framework such that the linear running time holds with high probability. Since there is very little elbow room in the running time when committed to linear running time, this extension is quite challenging and requires several new ideas and insights. See \secref{high_prob} for details.

\paragraph{What can be done without the net \& prune framework?} %
For the problems studied in this paper, an approximation algorithm with running time $O(n \log n)$ follows by using a $(1+\eps)$-\WSPD (well-separated pairs decomposition \cite{ck-dmpsa-95,h-gaa-11}), for some fixed $\eps$, to compute $O(n)$ relevant resolutions, and then using binary search with the decider to find the correct resolution. Even if the \WSPD is given, it is not clear, though, how to get a linear time algorithm for the problems studied without using the net \& prune framework. In particular, if the input has bounded spread, one can compute an $O(1)$-\WSPD in linear time, and then use the \WSPD inside the net \& prune framework to get a deterministic linear time algorithm, see \secref{bounded_spread}.

\paragraph{Paper organization.} %
We describe how to compute nets in \secref{prelim}, and how to remove faraway points efficiently in \secref{point_removal}.  We define the abstract framework, and describe and analyze the new approximation algorithm, in \secref{framework}.  We describe the applications in \secref{applications}.  We show how to modify the framework to achieve linear running time with high probability in \secref{high_prob}.  We conclude in \secref{conclusions}.

\section{Preliminaries}
\seclab{prelim}

\subsection{Basic definitions}

The \emphi{spread} of an interval $\Interval = [x,y] \subseteq \Re^{+}$ is $\SpreadX{\Interval} = y/x$.

\begin{definition}%
    \deflab{net}%
    For a point set $\PntSet$ in a metric space with a metric $\distChar$, and a parameter $r>0$, an $r$-\emphi{net} of $\PntSet$ is a subset $\cset \subseteq \PntSet$, such that
    \begin{inparaenum}[(i)]
        \item for every $\pnt,\pntA \in \cset$, $\pnt \neq \pntA$, we have that $\distX{\pnt}{\pntA} \geq r$, and
        \item for all $\pnt\in \PntSet$, we have that $\min_{\pntA\in \cset} \distX{\pnt}{\pntA} < r$.
    \end{inparaenum}
\end{definition}

Intuitively, an $r$-net represents $\PntSet$ in resolution $r$.

\begin{definition}%
    \deflab{grid_stuff}%
    For a number $\Delta > 0$ and a point $\pnt = (\pnt_1, \ldots, \pnt_d) \in \Re^d$, define $\Grid_\Delta(\pnt)$ to be the grid point $\pth[]{\floor{\pnt_1/\Delta}\! \Delta, \ldots, \floor{\pnt_d/\Delta}\! \Delta}$.

    The quantity $\Delta$ is the \emphi{width} or \emphi{sidelength} of the \emphi{grid} $\Grid_\Delta$. Observe that $\Grid_\Delta$ partitions $\Re^d$ into cubes, which are grid \emphic{cells}{cell}. The grid cell of $\pnt$ is uniquely identified by the integer point
    \begin{equation*}
        \gidX{\pnt} = \pth{ \floor{\pnt_1/\Delta}, \ldots,
           \floor{\pnt_d/\Delta} \bigr.}.
    \end{equation*}
\end{definition}

\begin{figure}[h!]
    \centerline{\includegraphics[scale=0.8]{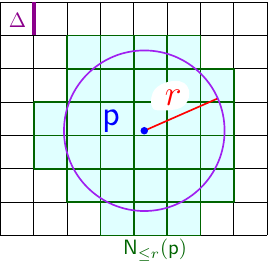}}
    \caption{}
    \figlab{neighborhood}
\end{figure}
\begin{definition}
    For a number $r \geq 0$, let $\GridNbr{\pnt}{r}$ denote the set of grid cells in distance $\leq r$ from $\pnt$, which is the \emphi{neighborhood} of $\pnt$. Note that the neighborhood also includes the grid cell containing $\pnt$ itself. If $\Delta =\Theta(r)$ then $\cardin{\GridNbr{\pnt}{r}} = \Theta\pth{(2+ \ceil{2r/\Delta})^d} = \Theta(1)$. See \figref{neighborhood}.
\end{definition}

\subsection{Computing nets quickly for a point set %
   in \TPDF{$\Re^d$}{Rd}}
\seclab{compNets}

There is a simple algorithm for computing $r$-nets.  Namely, let all the points in $\PntSet$ be initially unmarked.  While there remains an unmarked point, $\pnt$, add $\pnt$ to $\cset$, and mark it and all other points in distance $< r$ from $\pnt$ (i.e., we are scooping away balls of radius $r$).  By using grids and hashing, one can modify this algorithm to run in linear time.  The following is implicit in previous work \cite{h-cm-04}. We include it here for the sake of completeness\footnote{Specifically, the algorithm of Har-Peled \cite{h-cm-04} is considerably more complicated than that of \lemref{net}, and does not work in this setting, as the number of clusters it can handle is limited to $O\pth{n^{1/6}}$. \lemref{net} has no such restriction.} -- it was also described by the authors in \cite{ehr-fclb-12}.

\begin{lemma}%
    \lemlab{net}%
    Given a point set $\PntSet\subseteq \Re^d$ of size $n$ and a parameter $r>0$, one can compute an $r$-net for $\PntSet$ in $O(n)$ time.
\end{lemma}

\begin{proof}
    Let $\Grid$ denote the grid in $\Re^d$ with side length $\Delta = r/\pth{2\sqrt{d}}$.  First compute for every point $\pnt \in \PntSet$ the grid cell in $\Grid$ that contains $\pnt$; that is, $\gidX{\pnt}$.  Let $\gset$ denote the set of grid cells of $\Grid$ that contain points of $\PntSet$.  Similarly, for every cell $\cell \in \gset$ we compute the set of points of $\PntSet$ that it contains.  This task can be performed in linear time using hashing and bucketing, assuming the floor function can be computed in constant time. Specifically, store the $\gidX{\cdot}$ values in a hash table, and in constant time, hash each point into its appropriate bin.

    Scan the points of $\PntSet$ one at a time, and let $\pnt$ be the current point.  If $\pnt$ is marked, then move on to the next point.  Otherwise, add $\pnt$ to the set of net points, $\cset$, and mark it and each point $\pntA \in \PntSet$ such that $\distX{\pnt}{\pntA} < r$.  Since the cells of $\GridNbr{\pnt}{r}$ contain all such points, we only need to check the lists of points stored in these grid cells.  At the end of this procedure, every point is marked.  Since a point can only be marked if it is in distance $< r$ from some net point, and a net point is only created if it is unmarked when visited, this implies that $\cset$ is an $r$-net.

    As for the running time, observe that a grid cell, $c$, has its list scanned only if $c$ is in the neighborhood of some created net point.  As $\Delta = \Theta(r)$, there are only $O(1)$ cells which could contain a net point $\pnt$ such that $c\in \GridNbr{\pnt}{r}$.  Furthermore, at most one net point lies in a single cell since the diameter of a grid cell is strictly smaller than $r$.  Therefore, each grid cell had its list scanned $O(1)$ times.  Since the only real work done is in scanning the cell lists and since the cell lists are disjoint, this implies an $O(n)$ running time overall.
\end{proof}

Observe that the closest net point, for a point $\pnt \in \PntSet$, must be in one of its neighborhood's grid cells. Since every grid cell can contain only a single net point, it follows that in constant time per point of $\PntSet$, one can compute each point's nearest net point.  We thus have the following.

\begin{corollary}%
    \corlab{valid}%
    For a set $\PntSet \subseteq \Re^d$ of $n$ points, and a parameter $r > 0$, one can compute, in linear time, an $r$-net of $\PntSet$, and for each net point, the set of points of $\PntSet$ for which it is the nearest net point.
\end{corollary}

In the following, a \emphi{weighted point} is a point that is assigned a positive integer weight.  For any subset $S$ of a weighted point set $\PntSet$, let $\cardin{S}$ denote the number of points in $S$ and let $\WeightX{S} = \sum_{\pnt \in S} \WeightX{\pnt}$ denote the total weight of $S$.

In particular, \corref{valid} implies that for a weighted point set, one can compute the following quantity in linear time.

\begin{algorithm}[\net]%
    \AlgorithmAnchor{net}%
    \alglab{net}%
    Given a weighted point set $\PntSet \subseteq \Re^d$, let $\netX{r}{\PntSet}$ denote an $r$-net of $\PntSet$, where the weight of each net point $\pnt$ is the total sum of the weights of the points assigned to it.  We slightly abuse notation, and also use $\netX{r}{\PntSet}$ to designate the algorithm (of \corref{valid}) for computing this net, which has linear running time.
\end{algorithm}

\subsection{Identifying close and far points}
\seclab{point_removal}

For a given point $\pnt \in \PntSet$, let
\begin{equation}
    \nnX{\pnt}{\PntSet} =
    \arg \min_{\pntA \in \PntSet \setminus \brc{\pnt}}
    \distX{\pnt}{\pntA}%
    \qquad\text{ and } \qquad
    \ddX{\pnt}{\PntSet}  = \min_{\pntA \in \PntSet \setminus \brc{\pnt}}
    \distX{\pnt}{\pntA},%
    \eqlab{dd}
\end{equation}
denote the nearest neighbor of $\pnt$ in $\PntSet \setminus \brc{\pnt}$, and the distance to it, respectively.
The quantity $\ddX{\pnt}{\PntSet}$ can be computed (naively) in linear time by scanning the points.  For a set of points $\PntSet$, and a parameter $r$, let $\FarX{\PntSet}{r}$ denote the set of \emphi{$r$-far} points; that is, it is the set of all points $\pnt \in \PntSet$, such that the nearest-neighbor of $\pnt$ in $\PntSet\setminus \brc{\pnt}$ is at least distance $r$ away (i.e., $\ddX{\pnt}{\PntSet} \geq r$).  Similarly, $\CloseX{\PntSet}{r}$ is the set of \emphi{$r$-close} points; that is, all points $\pnt \in \PntSet$, such that $\ddX{\pnt}{\PntSet} < r$.

\begin{lemma}
    \AlgorithmAnchor{delete}%
    \lemlab{delete}%
    Given a weighted set $\PntSet \subseteq \Re^d$ of $n$ points, and a distance $\nRad > 0$, in $O\pth{n}$ time, one can compute the sets $\CloseX{\PntSet}{\nRad}$ and $\FarX{\PntSet}{\nRad}$.  Let $\delete(\nRad, \PntSet)$ denote the algorithm which computes these sets and returns $\CloseX{\PntSet}{\nRad}$.
\end{lemma}

\begin{proof}
    Build a grid where every cell has diameter $\nRad/c$, for some constant $c>1$.  Clearly, any point $\pnt\in \PntSet$ such that $\ddX{\pnt}{\PntSet}\geq \nRad$ (i.e., a far point) must be in a grid cell by itself.  Therefore, to determine all such ``far'' points, only grid cells with singleton points in them need to be considered.  For such a point $\pntA$, to determine if $\ddX{\pntA}{\PntSet} \geq \nRad$, one checks the points stored in all grid cells in distance $\leq \nRad$ from it. If $\pntA$ has no such close neighbor, we mark $\pntA$ for inclusion in $\FarX{\PntSet}{r}$.  By the same arguments as in \lemref{net}, the number of such cells is $O(1)$.  Again, by the arguments of \lemref{net} every non-empty grid cell gets scanned $O(1)$ times overall, and so the running time is $O(n)$. Finally, we copy all the marked (resp. unmarked) points to $\FarX{\PntSet}{\nRad}$ (resp. $\CloseX{\PntSet}{\nRad}$).  The algorithm $\delete(\nRad, \PntSet)$ then simply returns $\CloseX{\PntSet}{\nRad}$.
\end{proof}

\section{Approximating nice distance problems}
\seclab{framework}

\subsection{Problem definition and an example}

\begin{definition}
    Let $\PntSet$ and $\PntSetB$ be two sets of weighted points in $\Re^d$ (of the same weight).  The set $\PntSetB$ is a \emphi{$\Delta$-\translation} of $\PntSet$, if $\PntSetB$ can be constructed by moving each point of $\PntSet$ by a distance at most $\Delta$ (and not altering its weight). Formally, there is an onto mapping $f:\PntSet \rightarrow \PntSetB$, such that
    \begin{inparaenum}[(i)]
        \item For $\pnt \in \PntSet$, we have that $\distX{\pnt}{f(\pnt) }\leq \Delta$, and
        \item for any $\pntA \in \PntSetB$, we have that $\WeightX{\pntA}$ is the sum of the weights of all points $\pnt \in \PntSet$ such that $f(\pnt) = \pntA$.
    \end{inparaenum}
\end{definition}

Note that for a (potentially weighted) point set $\PntSetA$, $\netX{\Delta}{\PntSetA}$ is a $\Delta$-\translation of $\PntSetA$.

\begin{definition}
    \deflab{decider}%
    Given a set $X$ and a function $\target: X\rightarrow \Re$, a procedure $\decider$ is a \emphi{$\cDecider$-decider} for $\target$, if for any $x\in X$ and $r > 0$, $\decider(r, x)$ returns one of the following:
    \begin{inparaenum}[(i)]
        \item $\target(x)\in [\alpha, \cDecider \alpha]$, where $\alpha$ is some real number,
        \item $\target(x) < r$, or
        \item $\target(x) > r$.
    \end{inparaenum}
\end{definition}

\begin{definition}[$\cDecider$-\NDP]%
    \deflab{N_D_P}%
    An instance of a \emphi{$\cDecider$-\WellBehavedDistanceProblem{}} consists of a pair $(\PntSetA, \Context)$, where $\PntSetA \subseteq \Re^d$ is a set of $n$ distinct weighted points\footnote{The case when $\PntSetA$ is a multiset can also be handled. See \remref{multi}.}, and $\Context$ is the \emphi{context} of the given instance (of size $O(n)$) and consists of the relevant parameters for the problem\footnote{In almost all cases, the context is a set of integers, usually of constant size.}.  For a fixed \emphi{$\cDecider$-\NDP{}}, the task is to evaluate a function $\target(\PntSetA, \Context) \rightarrow \Re^+$, defined over such input pairs, that has the following properties:%
    \smallskip%
    \begin{compactitem}
        \item \PDecider: %
        There exists an $O(n)$ time $\cDecider$-decider for $\target{}$, for some \emph{constant} $\cDecider \geq 1$.

        \item \PLipschitz: %
        Let $\PntSetB$ be any $\Delta$-\translation of $\PntSetA$.  Then $\abs{\target(\PntSetA, \Context)-\target(\PntSetB, \Context)} \leq 2\Delta$.

        \item \PPrune: %
        If $\target(\PntSetA, \Context) < r$ then $\target(\PntSetA, \Context) = \target( \CloseX{\PntSetA}{r} , \Context ')$, where $\CloseX{\PntSetA}{r}$ is the set of $r$-close points, and $\Context'$ is an updated context which can be computed, in $O(n)$ time, by a procedure denoted \contextUpdate.
    \end{compactitem}
\end{definition}

\begin{remark}[$(1+\eps)$-\NDP]%
    If we are interested in a $(1+\eps)$-approx{}i{m}a\-t{i}on, then $\cDecider = 1+\eps$. In this case, we require that the running time of the provided $(1+\eps$)-decider (i.e., property \PDeciderRef{} above) is $O(n/\eps^c)$, for some constant $c \geq 1$. This condition holds for all the applications presented in this paper.

    Our analysis still holds even if the decider has a different running time, but the overall running time might increase by a factor of $\log (1/\eps)$, see \lemref{result_1} for details.
\end{remark}

\subsubsection{An example -- \TPDF{$k$}{k} center clustering}
\seclab{kCenter}

As a concrete example of an \NDP, consider the problem of $k$-center clustering.

\begin{problem}[\kCenter]
    Let $\PntSetA$ be a set of points in $\Re^d$, and let $k>0$ be an integer parameter.  Find a set of \emphi{centers} $\centers\subseteq \PntSetA$ such that the maximum distance of a point of $\PntSetA$ to its nearest center in $\centers$ is minimized, and $\cardin{\centers} = k$.

    Specifically, the function of interest, denoted by $\kCenterFunc(\PntSetA, k)$, is the radius of the optimal $k$-center clustering of $\PntSetA$.  We now show that \kCenter satisfies the properties of an \NDP.
\end{problem}

\begin{remark}
    Usually, one defines the input for \kCenter to be a set of unweighted points.  However, since the definition of \NDP requires a weighted point set, we naturally convert the initial input into a set of unit-weight points, $\PntSetA$.  Also, for simplicity we assume $\kCenterFunc(\PntSetA,k) \neq 0$, i.e. $\cardin{\PntSetA}>k$ (which can easily be checked).  Both assumptions would be used implicitly throughout the paper.
\end{remark}

\begin{lemma}
    \lemlab{decider_1}%
    The following relations between $\cardin{\netX{r}{\PntSetA}}$ and $\kCenterFunc(\PntSetA, k)$ hold:
    \begin{inparaenum}[(A)]
        \item If we have $\cardin{ \netX{r}{\PntSetA}} \leq k$ then $\kCenterFunc(\PntSetA, k)< r$.
        \item If $\cardin{\netX{r}{\PntSetA}}\allowbreak > k$ then $r\leq 2\kCenterFunc(\PntSetA, k)$.
    \end{inparaenum}
\end{lemma}
\begin{proof}
    (A) Observe that if $\cardin{\netX{r}{\PntSetA}}\leq k$, then, by definition, the set of net points of $\netX{r}{\PntSetA}$ is a set of $\leq k$ centers such that all the points of $\PntSetA$ are in distance strictly less than $r$ from these centers.

    (B) If $\cardin{\netX{r}{\PntSetA}}> k$ then $\PntSetA$ must contain a set of $k+1$ points whose pairwise distances are all at least $r$.  In particular, any solution to \kCenter with radius $< r/2$ would not be able to cover all these $k+1$ points with only $k$ centers.
\end{proof}

\begin{lemma}%
    \lemlab{k_center_WBDP}%
    For any instance $(\PntSetA, k)$ of \kCenter, $\kCenterFunc(\PntSetA, k)$ satisfies the properties of \defref{N_D_P}; that is, \kCenter is a $(4+\eps)$-\NDP, for any $\eps>0$.
\end{lemma}
\begin{proof}
    We need to verify the required properties, see \defref{N_D_P}.

    \begin{compactitem}[\ItemSpacing]
        \itemPDecider We need to describe a decision procedure for \kCenter clustering.  To this end, given a distance $r$, the decider first calls $\netX{r}{\PntSetA}$, see \algref{net}.  If we have $\cardin{\netX{r}{\PntSetA}} \leq k$, then by \lemref{decider_1} the answer ``$\kCenterFunc(\PntSetA, k)< r$'' can be returned.

        Otherwise, call $\netX{(2+\eps/2)r}{\PntSetA}$.  If $\cardin{\netX{(2+\eps/2)r}{\PntSetA}}\leq k$, then, by \lemref{decider_1}, we have $r/2 \leq \kCenterFunc(\PntSetA, k)< (2+\eps/2)r$ and ``$\target \in [r/2, (2+\eps/2)r]$'' is returned.  Otherwise $\cardin{\netX{(2+\eps/2)r}{\PntSetA}}> k$, and \lemref{decider_1} implies that the answer ``$r< \kCenterFunc(\PntSetA, k)$'' can be returned by the decider.

        \itemPLipschitz Observe that if $\PntSetB$ is a $\Delta$-\translation of $\PntSetA$, then a point and its respective center in a $k$-center clustering of $\PntSetA$ each move by distance at most $\Delta$ in the transition from $\PntSetA$ to $\PntSetB$. As such, the distance between a point and its center changes by at most $2\Delta$ by this process. This argument also works in the other direction, implying that the $k$-center clustering radius of $\PntSetA$ and $\PntSetB$ are the same, up to an additive error of $2\Delta$.

        \itemPPrune It suffices to show that if $\kCenterFunc(\PntSetA, k)< r$ then
        \begin{equation*}
            \kCenterFunc(\PntSetA, k) =
            \kCenterFunc\pth{\CloseX{\PntSetA}{r}, k-
               \cardin{\FarX{\PntSetA}{r}}}.
        \end{equation*}
        Now, if $\kCenterFunc(\PntSetA, k)< r$ then any point of $\PntSetA$ whose neighbors are all $\geq r$ away must be a center by itself in the optimal $k$-center solution, as otherwise it would be assigned to a center $\geq r> \kCenterFunc(\PntSetA, k)$ away.  Similarly, any point assigned to it would be assigned to a point $\geq r> \kCenterFunc(\PntSetA,k)$ away.  Therefore, for any point $\pnt\in \PntSetA$ whose nearest neighbor is at least distance $r$ away, $\kCenterFunc(\PntSetA, k)=\kCenterFunc(\PntSetA\setminus \brc{\pnt}, k-1)$.  Repeating this observation implies the desired result.  The context update (and computing $\CloseX{\PntSetA}{r}$ and $\FarX{\PntSetA}{r}$) can be done in linear time using \delete, see \lemref{delete}. $\qedhere$
    \end{compactitem}
\end{proof}

\subsection{A linear time approximation algorithm}
\seclab{linear}

\begin{figure}[t]%
    \AlgorithmAnchor{\si{ndp_Alg}}%
    \centerline{%
       \fbox{%
          \begin{nprogram}
              \ndpAlg{}$(\PntSetA, \Context)$:\+\\
              Let $\PntSetA_0=\PntSetA$, $\Context_0=\Context$ and $i=1$.\\ %
              \While {\True} \Do \+\\
              Randomly pick a point $\pnt$ from $\PntSetA_{i-1}$.%
              \linelab{gen_random}%
              \\
              $\nRad_i \leftarrow \ddX{\pnt}{\PntSetA_{i-1}}$.%
              \linelab{rad}\\%
              \medskip\CodeComment{%
                 // Next, estimate the value of %
                 $\target_{i-1} = \target(\PntSetA_{i-1}, \Context_{i-1})$}%
              \\%
              $res_> = \decider(\nRad_i, \PntSetA_{i-1}, \Context_{i-1})$\\
              $res_< = \decider(\constNet\nRad_i, \PntSetA_{i-1}, \Context_{i-1})$\\
              \>\tabright{\CodeComment{// $\constNet = 37$:
                    \corref{same_1}.}}\\
              \If {$res_> = $ ``$f_{i-1} \in [x,y]$'' } \Then%
              \Return ``$f\pth{\PntSetA, \Context} \in [x/2,2y]$''.%
              \linelab{bounded1x}%
              \\ %
              \If {$res_< = $ ``$f_{i-1} \in [x',y']$'' } \Then%
              \Return ``$f\pth{\PntSetA, \Context} \in [x'/2,2y']$''.%
              \linelab{bounded1y}%
              \\%
              \If $res_> = $ ``$\nRad_i\! <\! \target_{i-1}$'' %
              and $res_< = $\ ``$ \target_{i-1} \!<\!  \constNet \nRad_i$''
              \Then \+\\
              \Return ``$f\pth{\PntSetA, \Context} \in [\nRad_i/2, 2\constNet\nRad_i]$''.  %
              \linelab{bounded3}%
              \-\\
              \If $res_> = ``\target_{i-1} < \nRad_i$'' \Then\+\\
              $\PntSetA_i = \delete(\nRad_i, \PntSetA_{i-1})$%
              \>\tabright{ \CodeComment{%
                    // \textbf{Prune}}}%
              \linelab{delete}%
              \\{%
                 $\Context_i = \contextUpdate(\PntSetA_{i-1}, \Context_{i-1}, \PntSetA_i)$ %
              }\-\\
              \If {$res_< = ``\constNet\nRad_i < \target_{i-1}$''} \Then\+\\
              $\PntSetA_i = \netX{3\nRad_i}{\PntSetA_{i-1}}$%
              \>\tabright{ \CodeComment{// \textbf{Net}}} %
              \linelab{net}%
              \\{$\Context_i = \Context_{i-1}$ }
              \-\\
              $i=i+1$
          \end{nprogram}}}
    \caption{%
       The approximation algorithm. The implicit target value being approximated is $\target = \target\pth{ \PntSetA, \Context}$, where $\target$ is a $\cDecider$-\NDP. That is, there is a $\cDecider$-decider for $\target$, denoted by \decider, and the only access to the function $\target$ is via this procedure.%
    }
    \figlab{a_algorithm}%
\end{figure}%

We now describe the general algorithm which given a $\cDecider$-\NDP, $(\PntSetA, \Context)$, and an associated target function, $\target$, computes in linear time a $O(\cDecider)$ spread interval containing $\target(\PntSetA, \Context)$.  In the following, let \decider denote the given $\cDecider$-decider, where $\cDecider \geq 1$ is some parameter.  Also, let \contextUpdate denote the context updater associated with the given \NDP problem, see \defref{N_D_P}. Both \decider and \contextUpdate run in linear time.  The algorithm for bounding the optimum value of an \NDP is shown in \figref{a_algorithm}.

\begin{remark}
    \remlab{zero}%
    For the sake of simplicity of exposition, we assume that $\target(\PntSetA, \Context)\neq 0$.  Since all our applications have $\target(\PntSetA, \Context)\neq 0$, we choose to make this simplifying assumption.  In particular, we later show in \corref{same_1} that if initially $\target(\PntSetA, \Context)\neq 0$, then \ndpAlg preserves this property.
\end{remark}

\begin{remark}
    \remlab{multi}%
    Note that the algorithm of \figref{a_algorithm} can be modified to handle inputs where $\PntSetA$ is a multiset (namely, two points can occupy the same location), and we still assume (as done above) that $\target(\PntSetA, \Context)\neq 0$.  Specifically, it must be ensured that the distance computed in \lineref{rad} is not zero, as this is required for the call to \decider.  This issue can be remedied (in linear time) by first grouping all the duplicates of the point sampled in \lineref{gen_random} into a single weighted point (with the sum of the weights) before calling \lineref{rad}. This issue can be addressed by computing a net for a radius that is smaller than (say) half the smallest non-zero distance. How to compute this distance is described in \secref{closest_non_zero_dist}.
\end{remark}

\subsubsection{Analysis}

A \emphi{net iteration} of the algorithm is an iteration where \net gets called. A \emphi{prune iteration} is one where \delete gets called.  Note that the only other type of iteration is the one where the algorithm returns. %
In the following, $\CloseY{\PntSet}{\nRad}$ is defined analogously to $\CloseX{\PntSet}{\nRad}$ (see \secref{point_removal}).

\begin{lemma}
    \lemlab{remove}%
    Let $\PntSet$ be a point set. A $(2+\eps)\nRad$-net of $\PntSet$, for any $\eps>0$, can contain at most half the points of $\CloseY{\PntSet}{\nRad}$.
\end{lemma}
\begin{proof}
    Consider any point $\pnt$ in $\CloseY{\PntSet}{\nRad}$ which became one of the net points.  Since $\pnt \in \CloseY{\PntSet}{\nRad}$, a disk of radius $\nRad$ centered at $\pnt$ must contain another point $\pntA$ from $\CloseY{\PntSet}{\nRad}$ (indeed, $\pnt \in \CloseY{\PntSet}{\nRad}$ only if its distance from its nearest neighbor in $\PntSet$ is at most $\nRad$).  Moreover, $\pntA$ cannot become a net point since it is too close to $\pnt$.  Now, if we place a ball of radius $\nRad$ centered at each point of $\CloseY{\PntSet}{\nRad}$ which became a net point, then these balls will be disjoint because the pairwise distance between net points is $\geq (2+\eps)\nRad$.  Therefore, each point of $\CloseY{\PntSet}{\nRad}$ which became a net point can be charged to at least one point of $\CloseY{\PntSet}{\nRad}$ which did not make it into the net, such that no point gets charged twice.
\end{proof}

\begin{lemma}%
    \lemlab{time}%
    Given an instance $(\PntSetA, \Context)$ of an \NDP, %
    $\ndpAlg(\PntSetA, \Context)$ runs in expected $O(n)$ time.
\end{lemma}

\begin{proof}
    In each iteration of the while loop, the only non-trivial work done is in computing $\nRad_i$, the two calls to \decider, and the one call to either \net or \delete.  It has already been shown that all of these can be computed in $O(\cardin{\PntSetA_{i-1}})$ time.  Hence, the total running time for the algorithm is $O\pth{ \sum_{i=0}^{k-1} \cardin{\PntSetA_i} }$, where $k$ denotes the last (incomplete) iteration of the while loop.

    So consider the beginning of iteration $i<k$ of the while loop.  Let the points in $\PntSetA_{i-1}$ be labeled $\pnt_1, \pnt_2, \ldots, \pnt_m$ in increasing order of their nearest neighbor distance in $\PntSetA_{i-1}$.  Let $j$ be the index of the point chosen in \lineref{gen_random} and let $(\PntSetA_{i-1})^{\geq j}$ and $(\PntSetA_{i-1})^{\leq j}$ be the subset of the points with index $\geq j$ and index $\leq j$, respectively.  Now since a point is randomly picked in \lineref{gen_random}, with probability $\geq 1/2$, $j\in [m/4, 3m/4]$.  Let's call this event a \emphi{successful iteration}. %
    We have $\min\pth{\cardin{(\PntSetA_{i-1})^{\geq j}}, \cardin{(\PntSetA_{i-1})^{\leq j}}} \geq \cardin{\PntSetA_{i-1}}/4$ for a successful iteration.

    Since $i<k$ is not the last iteration of the while loop, either \delete or \net must get called.  If $\delete(\nRad_i, \PntSetA_{i-1})$ gets called (i.e. \lineref{delete}) then by \lemref{delete}, all of $(\PntSetA_{i-1})^{\geq j}$ gets removed.  So suppose \net gets called (i.e. \lineref{net}).  In this case \lemref{remove} implies that the call to \net removes at least $\cardin{(\PntSetA_{i-1})^{\leq j}}/2$ points.

    Therefore, for any iteration $i<k$, at least
    \begin{equation*}
        \nu_i = \min{}\allowbreak (\cardin{(\PntSetA_{i-1})^{\geq j}},
        \cardin{(\PntSetA_{i-1})^{\leq j}}/2 )
    \end{equation*}
    points get removed.  If an iteration is successful then $\nu_i \geq \cardin{\PntSetA_{i-1}}/8$.  In particular,
    \begin{math}
        \Ex{\nu_i \sep{ \cardin{\PntSetA_{i-1}}}} \geq \cardin{\PntSetA_{i-1}}/16.
    \end{math}
    Now, $\cardin{\PntSetA_i}\leq \cardin{\PntSetA_{i-1}}-\nu_i$ and as such $\Ex{\cardin{\PntSetA_i} \sep{ \cardin{\PntSetA_{i-1}}}}\leq (15/16) \cardin{\PntSetA_{i-1}}$.

    Therefore, for $0 < i < k$,
    \begin{align*}
      \Ex{ \Bigl. \cardin{\PntSetA_i}}%
      =%
      \Ex{ \Ex{\cardin{\PntSetA_i} \sep{ \PntSetA_{i-1}} \bigr.}
      \Bigr.}%
      \leq%
      \Ex{\frac{15}{16}\cardin{\PntSetA_{i-1}}}%
      =%
      \frac{15}{16} \Ex{ \cardin{\PntSetA_{i-1}}\Bigr.}.
    \end{align*}
    Hence, by induction on $i$,
    \begin{math}
        \Ex{ \cardin{ \PntSetA_i } \bigr. } \leq (15/16)^i \cardin{\PntSetA_0}
    \end{math}
    and so, in expectation, the running time is bounded by
    \begin{align*}
      O\pth{\Ex{\sum_{i=0}^{k-1} \cardin{\PntSetA_i}}}%
      =%
      O\pth{\sum_{i=0}^{k-1} \Ex{ \cardin{\PntSetA_i} \Bigr.}}%
      =%
      O\pth{\sum_{i=0}^{k-1} (15/16)^i \cardin{\PntSetA_0}}%
      =%
      O\pth{\cardin{\PntSetA_0} \Bigr. }. \qedhere
    \end{align*}%
\end{proof}

\paragraph{Correctness.}
The formal proof of correctness is somewhat tedious, but here is the basic idea: At every iteration, either far points are being thrown away (and this does not affect the optimal value), or we net the points. However, the net radius being used is always significantly smaller than the optimal value, and throughout the algorithm execution, the radii of the nets being used grow exponentially. As such, the accumulated error in the end is proportional to the radius of the last net computation before termination, which itself is also much smaller than the optimal value.

\medskip

Before proving that \ndpAlg returns a bounded spread interval containing $\target(\PntSetA_0, \Context_0)$, several helper lemmas will be needed.  For notational ease, in the rest of this section, we use $\target(\PntSetA_i)$ as shorthand for $\target(\PntSetA_i, \Context_i)$.

\begin{lemma}%
    \lemlab{double}%
    Suppose that \net is called in iteration $i$ of the while loop (i.e., \lineref{net} in \figref{a_algorithm}).  Then for any iteration $j>i$ we have, $\nRad_{j} \geq 3\nRad_{i}$.
\end{lemma}
\begin{proof}
    Consider the beginning of iteration $j$ of the while loop.  The current set of points, $\PntSetA_{j-1}$, is a subset of the net points of a $3\nRad_{i}$-net (it is a subset since \lineref{delete} and \lineref{net} might have been executed in between rounds $i$ and $j$).  Therefore, being a net, the distance between any two points of $\PntSetA_{j-1}$ is $\geq 3\nRad_{i}$, see \defref{net}.  In particular, this means that for any point $\pnt$ of $\PntSetA_{j-1}$, we have $\ddX{\pnt}{\PntSetA_{j-1}} \geq 3\nRad_{i}$.
\end{proof}

\begin{lemma}%
    \lemlab{translation}%
    For $i=1, \ldots, k$, we have %
    $\abs{\target(\PntSetA_i) - \target(\PntSetA_0)}\leq 9\nRad_i$.
\end{lemma}
\begin{proof}
    Let $I$ be the set of indices of the net iterations up to (and including) the $i$\th iteration. Similarly, let $\overline{I}$ be the set of iterations where \delete gets called.

    If \net was called in the $j$\th iteration, then $\PntSetA_j$ is a $3\nRad_{j}$-\translation of $\PntSetA_{j-1}$ and so by the \PLipschitzRef property, $\abs{\target(\PntSetA_j)- \target(\PntSetA_{j-1})} \leq 6\nRad_j$.  On the other hand, if \delete gets called in the $j$\th iteration, then $\target(\PntSetA_j) = \target(\PntSetA_{j-1})$ by the \PPruneRef property. Let $m = \max I$, we have that
    \begin{align*}
      \abs{\target(\PntSetA_i) - \target(\PntSetA_0)}%
      &\leq%
        \sum_{j =1}^i \abs{\Bigl. \target(\PntSetA_j)-
        \target(\PntSetA_{j-1})}\\
      &%
        =%
        \sum_{j \in I} \abs{\Bigl. \target(\PntSetA_j)-
        \target(\PntSetA_{j-1})} + \sum_{j \in \overline{I}}
        \abs{\Bigl. \target(\PntSetA_j)- \target(\PntSetA_{j-1})} \\%
      &\leq%
        \sum_{j \in I} 6\nRad_j + \sum_{j \in \overline{I}} 0%
        \leq %
        6 \nRad_m \sum_{j=0}^\infty \frac{1}{3^j}%
        \leq%
        9 \nRad_m%
        \leq %
        9 \nRad_i,
    \end{align*}
    by \lemref{double}.
\end{proof}

The following lemma testifies that the radii of the nets computed by the algorithm are always significantly smaller than the value we are trying to approximate.

\begin{lemma}%
    \lemlab{induction}%
    For any iteration $i$ of the while loop such that \net gets called, we have $\nRad_i\leq \target(\PntSetA_0)/\constC$, where $0<\constC = \constNet-9$.
\end{lemma}

\begin{proof}
    The proof will be by induction.  Let $m_1, \dots m_t$ be the indices of the iterations of the while loop in which \net gets called.  For the base case, observe that in order for \net to get called, we must have $\constC \nRad_{m_1}< \constNet \nRad_{m_1}< \target(\PntSetA_{m_1-1})$.  However, since this is the first iteration in which \net is called, it must be that $\target(\PntSetA_0) = \target(\PntSetA_{m_1-1})$ (since for all previous iterations \delete must have been called).

    So now suppose that $\nRad_{m_j} \leq \target(\PntSetA_0)/\constC$ for all $m_j<m_i$.  If a call to \net is made in iteration $m_i$ then again $\constNet \nRad_{m_i}< \target(\PntSetA_{(m_i)-1}) = \target(\PntSetA_{m_{(i-1)}})$.  Thus, by \lemref{translation} and induction, we have
    \begin{align*}
      \nRad_{m_i}%
      &<%
        \frac{\target\pth{\PntSetA_{m_{(i-1)}}}}{\constNet}%
        \leq%
        \frac{\target(\PntSetA_0)+9\nRad_{m_{(i-1)}}}{\constNet}%
        \leq%
        \frac{\target(\PntSetA_0) +
        9\target(\PntSetA_0)/\constC}{\constNet}%
      \\&
      =%
      \frac{1 + 9/\constC}{\constNet} \target(\PntSetA_0)%
      =%
      \frac{\target(\PntSetA_0)}{\constC},
    \end{align*}
    if $\ds \constC = \frac{\constNet}{1+9/\constC}$. This in turn is equivalent to $\ds \constC + 9 = \constNet$, which is true by definition.
\end{proof}

Setting $\constNet = \constNetVal$, results in $\constC = \constCVal$, and by \lemref{induction}, for all $i$ that correspond to a net iteration, $\nRad_i\leq \target(\PntSetA_0)/\constCVal$.  By \lemref{translation}, for any net iteration $i$, we have
\begin{align*}
  \abs{\Bigl. \target(\PntSetA_i) - \target(\PntSetA_0)} \leq
  9\nRad_i \leq \target(\PntSetA_0) / 3 .
\end{align*}
In particular, we conclude that $\abs{\target(\PntSetA_i) - \target(\PntSetA_0)} \leq \target(\PntSetA_0) / 3 $ for any iteration $i$.  We thus get the following.

\begin{corollary}%
    \corlab{same_1}%
    For $\constNet \geq \constNetVal$, and any $i$, we have:
    \begin{compactenum}[\qquad(A)]
        \item $(2/3)\target(\PntSetA_0)$ $\leq \target(\PntSetA_{i}) \leq (4/3)\target(\PntSetA_0)$.

        \smallskip

        \item If $\target(\PntSetA_{i}) \in [x,y]$ then $\target(\PntSetA_0) \in [(3/4)x, (3/2)y]$ $\subseteq [x/2, 2y]$.

        \smallskip

        \item If $\target(\PntSetA_0)>0$ then $\target(\PntSetA_i) > 0$.
    \end{compactenum}
\end{corollary}

\begin{lemma}%
    \lemlab{summary_1}%
    For $\constNet\geq \constNetVal$, given an instance $(\PntSetA, \Context)$ of a $\cDecider$-\NDP, the algorithm $\ndpAlg(\PntSetA, \Context)$ returns an interval $[x',y']$ containing $\target(\PntSetA, \Context)$, where $\SpreadX{ \bigl. [x',y']} \leq 4 \max \pth{ \cDecider, \constNet}$.
\end{lemma}

\begin{proof}
    Consider the iteration of the while loop at which \ndpAlg terminates, see \figref{a_algorithm}.  If \lineref{bounded1x} or \lineref{bounded1y} get executed at this iteration, then the interval $[x,y]$ was computed by the $\cDecider$-decider, and has spread $\leq \cDecider$. As such, by \corref{same_1}, the returned interval $[x',y'] = [x/2, 2y]$ contains the optimal value, and its spread is $\leq 4\cDecider$.

    A similar argument holds if \lineref{bounded3} gets executed. Indeed, the returned interval contains the desired value, and its spread is $4 \constNet$.
\end{proof}

\subsection{The result}

In \defref{N_D_P} required that the decision procedure runs in $O(n)$ time.  However, since our analysis also applies to decision procedures with slower running times, we first state the main result in these more general settings.

\begin{lemma}%
    \lemlab{result_1}%
    Given an instance of a $\cDecider$-\NDP defined by a set of $n$ points in $\Re^d$, one can get a $\cDecider$-approximation to its optimal value, in $O\pth{ T(n)}$ expected time, where $T(n)$ is the running time of the \decider procedure, and $\cDecider \geq 3/2$.

    If $\cDecider = 1+\eps$, for some $\eps \in (0,1)$, (i.e., \decider is $(1+\eps)$-decider), then one can compute a $(1+\eps$)-approximation to the optimal value, and the expected running time is $O\pth{ T(n) \log \eps^{-1}}$.
\end{lemma}

\begin{proof}
    Let $(\PntSetA, \Context)$ be the given $\cDecider$-\NDP instance, and let \decider be the corresponding $\cDecider$-decider.  By \lemref{time} and \lemref{summary_1}, in expected $O(n)$ time, one can get a bounded spread interval $[\gamma, c\gamma]$, for some constant $c \geq 1$ such that $c =O(\cDecider)$, such that $\target = \target(\PntSetA, \Context)\in [\gamma,c\gamma]$.  If $c \leq \cDecider$ then we are done. Otherwise, perform a binary search over this interval.

    Specifically, for $i=0, 1,\dots, m=\floor{\log_\cDecider c }$, let $\gamma_i = \cDecider^i \gamma$ and let $\gamma_{m+1}= c \gamma$.  Now perform a binary search over the $\gamma_i$'s using \decider.  If any of the calls returns an interval $[x, \cDecider x]$ that contains the optimal value, then $\target\leq \cDecider x \leq \cDecider\target$ and so $\cDecider x$ can be returned as the desired $\cDecider$-approximation.  Otherwise, if $\target < \gamma_i$ or $\target >\gamma_i$, the binary search moves left or right, respectively.  In the end, the search ends up with an interval $(\gamma_i, \gamma_{i+1}) = \pth{\gamma_i, \cDecider \gamma_{i}} $ which must contain the optimal value, and again, its spread is $\cDecider$, and it thus provides the required approximation.

    The running time is dominated by the calls to the decider procedure (we are assuming here that $T(n) = \Omega(n)$). Clearly, the number of calls to the decider performed by this algorithm is $O( \log m) = O\pth{ \log \log_\cDecider 4\cDecider } = O(1)$, if $\cDecider \geq 3/2$. Otherwise, if $\cDecider = 1+\eps$, then
    \begin{align*}
      O( \log m)%
      &=%
        O\pth{ \log \log_\cDecider 4\cDecider }%
        =%
        O\pth{ \ln \pth{ 1 + \frac{\ln 4}{\ln (1+\eps)}}}%
      \\&
      =%
      O\pth{ \log \pth{ 1 + \frac{1}{\eps}}}%
      =%
      O \pth{ \log \frac{1}{\eps}},
    \end{align*}
    since $\ln (1+\eps) \geq \ln e^{\eps/2} = \eps/2$, for $\eps \in (0,1/2)$, as can be easily verified.
\end{proof}

The following is a restatement of the main result in the settings of \defref{N_D_P}, where it was assumed $T(n) = O(n)$, or $T(n) = O(n/\eps^d)$ when $\cDecider = 1+\eps$.  In the latter case, one can remove the $\log(1/\eps)$ factor from the running time.

\begin{theorem}%
    \thmlab{result_eps}%
    Given an instance of a $\cDecider$-\NDP defined by a set of $n$ points in $\Re^d$, one can get $\cDecider$-approximate the optimal value, in expected $O\pth{ n}$ time, assuming $\cDecider \geq 3/2$.

    For the case $\cDecider = (1+\eps) \in (1,2)$, given an $(1+\eps)$-decider, with running time $O(n/\eps^c)$, one can $(1+\eps$)-approximate, in $O\pth{n/\eps^c}$ expected time, the given $(1+\eps)$-\NDP, where $c \geq 1$ is some constant.
\end{theorem}

\begin{proof}
    \lemref{result_1} readily implies the first half, as we are now assuming $T(n) = O(n)$. As for the second half, \lemref{result_1} implies that one can get a $(1+\eps)$-approximation in expected $O(n\log(1/\eps)/\eps^c)$ time.  However, by using the same procedure as in the proof of \lemref{result_1} but with exponentially decreasing values for $\eps$ in the binary search, a factor of $\log(1/\eps)$ can be avoided.  This is a standard idea and was also used by Aronov and Har-Peled \cite{ah-adrp-08}.

    Let $\decider_{\eps}$ be our $(1+\eps)$-decider.  If $\eps$ is set to any constant (say $1/2$), then by \lemref{time} and \lemref{summary_1}, in expected linear time, one can get a bounded spread interval $[\gamma, c\gamma]$, for some $c=O(\cDecider)$, such that $\target = \target(\PntSetA, \Context)\in \Interval_0 = [\gamma,c\gamma]$.

    Set $c= 1+\eps_0$ and $i=1$. The algorithm now proceeds in rounds, doing the following in the $i$\th round:
    \begin{compactenum}[\quad(A)]
        \item Assume that at the beginning of the $i$\th iteration, we know that
        \begin{align*}
          f \in \Interval_{i-1} = \pbrc[]{x_{i-1}, y_{i-1}},%
          \qquad \text{ where } \qquad%
          y_{i-1} = \pth[]{1+\eps_{i-1}} x_{i-1}.
        \end{align*}
        If $\eps_{i-1} \leq \eps$, then we found the desired approximation, and the algorithm stops.

        \item Set $\eps_{i}=\sqrt{1+\eps_{i-1}} - 1$ and $m_i = \pth{1+\eps_i}x_{i-1}$.
        \item $R_i \leftarrow \decider_{\eps_i}(m_i, \PntSetA, \Context)$, see \defref{decider}.

        \item There are three possibilities:
        \begin{compactenum}[(i)]
            \item If $R_i = $ ``$f<m_i$'', then set $\Interval_i \leftarrow [x_{i-1}, m_i]$.
            \item If $R_i = $ ``$f>m_i$'', then set $\Interval_i \leftarrow [m_i, y_{i-1} ]$.
            \item If $R_i = $ ``$f \in \Interval = [z, (1+\eps_i)z]$'', then set $\Interval_i \leftarrow \Interval$.
        \end{compactenum}
        \item $i \leftarrow i+1$.
    \end{compactenum}
    \medskip%
    In each round, the algorithm computes an interval of spread $y_i/x_i = 1+\eps_i$ that contains $f$, and
    \begin{align*}
      \eps_{i}%
      =%
      \sqrt{1+\eps_{i-1}} - 1%
      =%
      \frac{1+\eps_{i-1} - 1}{\sqrt{1+\eps_{i-1}} + 1} %
      \leq %
      \frac{\eps_{i-1}}{2}.
    \end{align*}

    Since the main bottleneck in each iteration is calling $\decider_{\eps_i}$, which runs in $O(n/\eps_i^c)$ time, the total running time is bounded by,
    \begin{align*}
      \sum_{i=1} O\pth{n/\eps_i^c} = O(n/\eps^c),
    \end{align*}
    since the sum behaves like a geometric series and the last term is $O(n/\eps^c)$.
\end{proof}

\subsubsection{A deterministic algorithm for the bounded %
   spread case}
\seclab{bounded_spread}%

The random sampling in the algorithm of \thmref{result_eps} can be removed if the spread of $\PntSet$ is polynomially bounded. Indeed, then one can snap the points to a grid, so that the points have integer coordinates with polynomially bounded values. A compressed quadtree and an $O(1)$-\WSPD of the snapped point set can be computed in linear time \cite{h-gaa-11,c-wspdl-08}. Given such a \WSPD, one can compute for each point its (approximate) nearest neighbor, and compute the median distance in this list in linear time. Using this value in each iteration of \ndpAlg{} (instead of picking a random point and computing its nearest neighbor distance) results in an algorithm with linear running time. Note that the resulting algorithm is deterministic.

\section{Applications}
\seclab{applications}

We now show that \thmref{result_eps} can be applied to a wide array of problems.  In order to apply it, one needs to show that the given problem meets the requirements of an \NDP. Often, the inputs to the problems considered are unweighted point sets.  As such, when we say such a problem is an \NDP, it is actually defined for unit-weighted point sets.  We also continue making the simplifying assumption that the optimum value is strictly positive (see \remref{zero}).
\subsection{\TPDF{$k$}{k}-center clustering}
\seclab{k_center}

Since computing a $k$-center clustering is an \NDP problem (\lemref{k_center_WBDP}), plugging this into \thmref{result_eps} immediately yields a constant factor approximation to $k$-center in linear time. It is easy to convert such an approximation to a $2$-approximation using a grid, see Har-Peled \cite[Lemma 6.5]{h-cm-04}. Thus, we get the following.

\begin{theorem}
    \thmlab{k_center_2}%
    Given a set $\PntSet$ of $n$ points in $\Re^d$, and a parameter $k$, $1\leq k \leq n$, one can compute a $2$-approximation to the optimal $k$-center clustering of $\PntSet$ in (expected) linear time. More precisely, the expected running time is $O(n + k \log k)$.
\end{theorem}

A result similar to \thmref{k_center_2} was already known for the case $k=O\pth{n^{1/3}/\log n}$ \cite{h-cm-04}, and the above removes this restriction. In addition, the new algorithm and its analysis are both simpler than the previous algorithm.

\subsection{The \TPDF{$k$}{k}\th smallest distance}
\seclab{k_th_distance}

\begin{claim}%
    \clmlab{change}%
    Let $S$ and $S'$ be subsets of $\Re$ of size $n$, such that $S'$ is obtained by taking each value in $S$ and increasing or decreasing it by less than $\Delta$.  Let $v$ and $v'$ be the $k$\th smallest values in $S$ and $S'$, respectively.  Then $\abs{v-v'}\leq \Delta$.
\end{claim}
\begin{proof}
    Suppose for contradiction that $\abs{v- v'}> \Delta$.  If $v' - v > \Delta$ then $S'$ has at least $n-k+1$ values strictly larger than $v+\Delta$, which implies $S$ has at least $n-k+1$ values strictly larger than $v$.  Similarly, if $v-v' > \Delta$ then $S'$ has at least $k$ values strictly smaller than $v-\Delta$, which implies $S$ has at least $k$ values strictly smaller than $v$.
\end{proof}

\begin{lemma}%
    \lemlab{k_dist}%
    Let $\PntSetA$ be a weighted point set in $\Re^d$, and let $k>0$ be an integer parameter.  Let $\binom{\PntSetA}{2}$ denote the multi-set of pairwise distances determined by $\PntSetA$.%
    \footnote{In particular a point of weight $m$ is viewed as $m$ unit weight points when determining the values of $\binom{\PntSetA}{2}$.  For simplicity, we assume that the $k$\th distance in $\PntSetA$ is not zero (i.e., it is determined by two distinct points).} %
    Given an instance $(\PntSetA, k)$ the \emphi{\kthDist} problem asks you to output the $k$\th smallest distance in $\binom{\PntSetA}{2}$. Given such an instance, one can $(1+\eps)$-approximate the $k$\th smallest distance in (expected) time $O\pth{n/\eps^d}$.
\end{lemma}
\begin{proof}
    Let $f(\PntSetA, k)$ be the function that returns the $k$\th smallest distance.  We prove that this function is an $(1+\eps)$-\NDP (see \defref{N_D_P}).

    \begin{compactitem}[\ItemSpacing]
        \itemPDecider Given $r$ and $\eps$, build a grid where every cell has diameter $\delta = \eps r/8$, and store the points of $\PntSetA$ in this grid.  Now, for any non-empty grid cell $\cell$, let $\WeightX{\cell}$ be the total weight of points in $\cell \cap \PntSetA$, and register this weight with all the grid cells in distance at most $r$ from it, i.e. all cells in $\GridNbr{\cell}{r}$ (where $\GridNbr{\cell}{r}$ is defined analogously for grid cells as $\GridNbr{\pnt}{r}$ was for points in \defref{grid_stuff}).  Let $\NWeightX{\cell}$ denote the total weight registered with a cell $\cell$ (which includes its own weight).  Any point in $\cell$ determines $\NWeightX{\cell}-1$ distances to other points in $\GridNbr{\cell}{r}$.  Therefore, the total number of distances between points in $\cell$ and points in $\GridNbr{\cell} {r}$ is $\WeightX{\cell} \pth{ \bigl. \NWeightX{\cell} - 1}$.  Summing this over all cells (and dividing by $2$ to avoid double counting) gives the quantity
        \begin{align*}
          S = \sum_{\cell, \WeightX[]{\cell}\ne 0}
          \frac{\WeightX{\cell}}{2} \pth{ \Bigl. \NWeightX{\cell} -
          1}.
        \end{align*}
        Note that the count $S$ includes all distances which are $\leq r$, and no distances $> r+2 \delta$.  So let the desired $k$\th distance be denoted by $\ell_k$.  Thus, if $k \leq S$ then $\ell_k \leq r + 2\delta = (1+\eps/4)r$.  Similarly, if $k > S$ then $\ell_k > r$.  Therefore, to build a $(1+\eps)$-decider for distance $r$, we run the above counting procedure on $r_1 = r/(1+\eps/3)$ and $r_2=r$, and let $S_1$ and $S_2$ be the two counts computed, respectively. There are three cases: %
        \smallskip
        \begin{compactenum}[\quad (I)]
            \item If $k\leq S_1$ then $\ell_k \leq (1+\eps/4)r/(1+\eps/3) < r$, and return this result.
            \item If $k\leq S_2$ then $\ell_k \leq (1+\eps/4)r_2 = (1+\eps/4)r$, then $\ell_k \in [r/(1+\eps/3),(1+\eps/4)r]$, and this interval has spread $(1+\eps/4)(1+\eps/3) < 1+ \eps$. The decider returns that this interval contains the desired value.

            \item The only remaining possibility is that $\ell_k > r_2$, and the decider returns this.
        \end{compactenum}
        \smallskip

        The running time of this decision procedure is $O\pth{n/\eps^d}$.

        \smallskip

        \itemPLipschitz Since the distance between any pair of points changes by at most $2\Delta$ in a $\Delta$-\translation, the Lipschitz condition holds by \clmref{change}.

        \itemPPrune By assumption, the $k$\th smallest distance is determined by two distinct weighted points $\pnt$ and $\pntA$.  Clearly these points are not in $\FarX{\PntSetA}{r}$ since $\distX{\pnt}{\pntA} = \target(\PntSetA, k) <r$.  So consider any point $\pntB \in \FarX{\PntSetA}{r}$.  Since removing $\pntB$ does not remove the distance $\distX{\pnt}{\pntA}$ from the set of remaining distances, all that is needed is to show how to update $k$. Clearly, $\pntB$ contributes $\WeightX{\pntB} \times \WeightX{\PntSetA \setminus \brc{\pntB}}$ distances, all of value $\geq r$, to $\binom{\PntSetA}{2}$, and $\binom{\WeightX{\pntB}}{2}$ pairwise distances of value zero. Thus, after removing $\pntB$ from $\PntSetA$, the new context is $k - \binom{\WeightX{\pntB}}{2}$.  $\qedhere$
    \end{compactitem}
\end{proof}

The algorithm of \lemref{k_dist} also works (with minor modifications) if we are interested in the $k$\th distance between two sets of points (i.e., the bipartite version).
\begin{corollary}
    Given two sets $\PntSet$ and $\PntSetB$ of points in $\Re^d$, of total size $n$, and parameters $k$ and $\eps>0$, one can $(1+\eps)$-approximate, in expected $O(n/\eps^d)$ time, the following:
    \begin{compactenum}[\quad(A)]
        \item The $k$\th smallest bichromatic distance. Formally, this is the smallest $k$\th number in the multiset $X = \brc{ \distX{\pnt}{\pntA} \sep{\pnt \in \PntSet, \pntA \in \PntSetA}}$.
        \item The closest bichromatic pair between $\PntSet$ and $\PntSetB$.
    \end{compactenum}
\end{corollary}

\subsection{The \TPDF{$k$}{k}\th smallest \TPDF{$m$}{m}-nearest neighbor distance}
\seclab{k_smallest_m_n_n}

For a set $\PntSet = \brc{\pnt_1,\ldots, \pnt_n} \subseteq \Re^d$ of $n$ points, and a point $\pnt \in \PntSet$, its \emphi{$m$\th nearest neighbor} in $\PntSet$ is the $m$\th closest point to $\pnt$ in $\PntSet\setminus \brc{\pnt}$. In particular, let $\ddmX{m}{\pnt}{\PntSet}$ denote this distance.  Here, consider the multiset of these distances defined for each point of $\PntSet$; that is,
\begin{math}
    X = \brc{ \ddmX{m}{\pnt_1}{\PntSet}, \ldots, \ddmX{m}{\pnt_n}{\PntSet}}.
\end{math}
Interestingly, for any $m$, we can approximate the $k$\th smallest number in this set in linear time.

To use \ndpAlg, one needs to generalize the problem to the case of weighted point sets. Specifically, a point $\pnt$ of weight $\WeightX{\pnt}$ is treated as $\WeightX{\pnt}$ distinct points at the same location.  As such, the set $X$ from the above is actually a multiset containing $\WeightX{\pnt}$ copies of the value $\ddmX{m}{\pnt}{\PntSet}$.  In particular, if $\WeightX{\pnt}> m$ then $\ddmX{m}{\pnt}{\PntSet}=0$.

\begin{theorem}%
    \thmlab{k_th_m_nn}%
    Let $\PntSetA$ be a set of $n$ weighted points in $\Re^d$, and $m,k, \eps$ parameters. Then one can $(1+\eps)$-approximate, in expected $O(n /\eps^d)$ time, the $k$\th smallest $m$-nearest neighbor distance in $\PntSetA$. Formally, the algorithm $(1+\eps)$-approximates the $k$\th smallest number in the multiset $X = \brc{ \ddmX{m}{\pnt_1}{\PntSetA}, \ldots, \ddmX{m}{\pnt_n}{\PntSetA}}$
\end{theorem}
\begin{proof}
    We establish the required properties, see \defref{N_D_P}.
    \begin{compactitem}[\ItemSpacing]
        \itemPDecider Let $f(\PntSetA,k,m)$ denote the desired quantity. We first need a $(1+\eps)$-decider for this problem. To this end, given $r,k, m, \eps$ and $\PntSetA$ as input, create a grid with cells of diameter $\eps r/4$, and mark for each point of $\PntSetA$ all the grid cells in distance at most $r$ from it. Each non-empty grid cell has a count of the total weight of the points in distance at most $r$ from it. Thus, each point of $\PntSetA$ can compute the total weight of all points that are approximately at a distance of at most  $r$  from it.  Now, a point $\pnt \in \PntSetA$, can decide in constant time if (approximately) $\ddmX{m}{\pnt}{\PntSetA} \leq r$.  If the number of points which declare $\ddmX{m}{\pnt}{\PntSetA} \leq r$ (where a point $\pnt$ is counted $\WeightX{\pnt}$ times) is greater than $k$, then the distance $r$ is too large, and if it is smaller than $k$, then $r$ is too small. Being slightly more careful about the details (as was done in the proof of \lemref{k_dist}), one can verify this leads to a $(1+\eps)$-decider for this problem.

        \itemPLipschitz Clearly the Lipschitz property holds in this case, as $\Delta$-\translation only changes inter-point distances by at most an additive term of $2\Delta$.

        \itemPPrune We need to show that $k$ can be updated properly for a weighted point $\pnt$ whose nearest neighbor is further away than $\target(\PntSetA, k, m)$.  If $\WeightX{\pnt} > m$ then we are throwing away $\WeightX{\pnt}$ points, all with $\ddmX{m}{\pnt}{\PntSetA} =0$, and so $k$ should be update to $k-\WeightX{\pnt}$ when $\pnt$ is removed. Similarly, if $\WeightX{\pnt} \leq m$ then all the points that $\pnt$ corresponds to have $\ddmX{m}{\pnt}{\PntSetA}$ larger than the threshold, and $k$ does not have to be updated.
    \end{compactitem}
    \smallskip%
    Plugging this into \thmref{result_eps} implies the desired result.
\end{proof}

\begin{remark}
    \remlab{h_really}%
    \thmref{k_th_m_nn} can be easily extended to work in the bichromatic case. That is, there are two point sets $\PntSet$ and $\PntSetB$, and we are interested in the $m$\th nearest neighbor of a point $\pnt \in \PntSet$ in the set $\PntSetB$. It is easy to verify that the same time bounds of \thmref{k_th_m_nn} hold in this case.

    In particular, setting $k = \cardin{\PntSet}$ and $m=1$ in the bichromatic case, the computed distance will be the minimum radius of the balls needed to be placed around the points of $\PntSetB$ to cover all the points of $\PntSet$ (or vice versa).
\end{remark}

\subsubsection{Hausdorff distance between point sets}
\seclab{Hausdorff}

\begin{definition}
    Given finite sets $R$ (red) and $B$ (blue) in $\Re^d$, consider the furthest distance of a point in $R$ from a point of $B$:
    \begin{equation*}
        \dDirY{R}{B} = \max_{r \in R} \dmY{r}{B},
    \end{equation*}
    where $\dmY{r}{B} = \min_{b \in B} \dY{b}{r}$ (note the confusing flip from $\max$ to $\min$). This is the maximal distance a red point has to travel till arriving at a blue point.  The \emphi{Hausdorff distance} between $R$ and $B$ is
    \begin{equation*}
        \dHDY{R}{B}
        =%
        \max\bigl( \dDirY{R}{B}, \dDirY{B}{R} \bigr).
    \end{equation*}
\end{definition}
The Hausdorff distance is a metric, useful in measuring similarity between sets.

\begin{problem}[Hausdorff distance]
    Given two sets of points $R$ and $B$ in $\Re^d$ compute $\dHDY{R}{B}$.
\end{problem}

A careful reading of \remref{h_really} reveals that this is precisely the problem it solves. We sketch an explicit description of this algorithm in detail here for completeness.

The algorithm maintains a net for $R$ and $B$, where each point marks whether its clients (i.e., all the points that got moved into it) were all originally red (i.e., from $R$), blue (i.e., from $B$), or mixed color. Such a set of points is \emphi{labeled}.

\begin{lemma}
    \lemlab{decider_Hausdorff}%
    The input is a labeled point set $P$ in $\Re^d$ of total size $n$, let $R$ (resp. $B$) be all the points in $P$ that are labeled red (resp., blue), where a mixed color point appears in both sets. Let $r$ and $\eps > 0$ be additional input parameters. Then, the algorithm described below, outputs, in $O( (1+1/\eps^d) n )$ time, one of the following:
    \begin{compactenum}[(I)]
        \item $\dHDY{R}{B} \leq r$,
        \item $\dHDY{R}{B} > r$, or
        \item returns an interval $J = [\alpha, (1+\eps)\alpha]$ such that $\dHDY{R}{B} \in J$.
    \end{compactenum}
\end{lemma}
\begin{proof}
    (If $\eps > 1$, the algorithm is executed with $\eps=1$.)  Let $\delta = \eps r / 8$.  Let $\Grid_B$ be a grid with the cells having diameter $\delta$, and throw the points of $B$ into this grid, where each non-empty grid cell marks all the grid cells in distance $\leq r$ from it, as being close to a ``blue'' point of $B$. Next, stream the points of $R$ through the grid $\Grid_B$ -- if any point falls in an unmarked cell, its distance from $B$ is larger than $r$, and return ``$\dHDY{R}{B} > r$''.  Next, do the same but reversing the roles of the colored point sets.

    The above replaces every point of $R$ and $B$ by the grid cell containing it, and checking directly on the grid, if these grid cells are within a distance $\leq r$ from each other, and it only shrinks the Hausdorff distance.

    Thus, if both stages succeeded, then $\dHDY{R}{B} \leq r + 2\delta \leq (1+\eps/4)r$. We repeat the above process, but now with distance $r' =\tfrac{r}{1+\eps/4}$. If again the answer is that $\dHDY{R}{B} \leq (1+\eps/4)r' = r$, then the algorithm returns that ``$\dHDY{R}{B} \leq r$''.

    Otherwise, the interval $J = \bigl((1-\eps/4)r, (1+\eps/4)r\bigr)$ contains $\dHDY{R}{B}$, and as $\tfrac{1+\eps/4}{1-\eps/4} \leq 1+\eps$, for $\eps \in (0,1]$, the algorithm returns the interval $J$.
\end{proof}

\begin{lemma}%
    \lemlab{Hausdorff_is_nice}%
    For any labeled point set $\PS = R \cup B$ in $\Re^d$, we have that $\target(R \cup B) = \dHDY{R}{B}$ is a \defrefY{N_D_P}{$(1+\eps)$-\NDP} for any $\eps > 0$.
\end{lemma}
\begin{proof}
    We verify the required properties, see \defref{N_D_P}.

    \smallskip%
    \PDecider{}: This is the algorithm of \lemref{decider_Hausdorff}.  \smallskip%

    \PLipschitz: Observe that if $R', B'$ are $\Delta$-drifts of $R$ and $B$, respectively, then for any two points $r' \in R', b' \in B'$, and their original points $r \in R$ and $b \in B$, we have that $\dY{r'}{b'} \leq \dY{r}{b} + 2\Delta$ and $\dY{r}{b} \leq \dY{r'}{b'} + 2\Delta$. Namely, we have that
    \begin{math}
        \cardin{ \dHDY{R}{B} - \dHDY{R'}{B'}} \leq 2\Delta.
    \end{math}

    \PPrune: We need to show that if $\dHDY{R}{B} < r$, then then $\dHDY{R_l}{ B_l} < r$, where $R_l \cup B_l$ is the residual set after throwing away all the $r$-far points in the set $R \cup B$. Consider any $r$-far point $p \in R \cup B$. Observe that $p$ must be mixed --- if it is monochromatic, then it is too far to be in Hausdorff distance $<r$ from the other color. So all the far points are mixed. Furthermore, since they are far, they can not serve as the destination of any outside point when computing the Hausdorff distance. Thus, they can be safely thrown away without changing the distance.
\end{proof}

Plugging the above into \thmref{result_eps}, we get the following.
\begin{theorem}
    \thmlab{h_really_2}%
    Given two sets of points $R$ and $B$ of total size $n$ in $\Re^d$, and a parameter $\eps > 0$, the above algorithm computes, in expected $O( (1+\tfrac{1}{\eps^d})n)$ time, a distance $\ell$, such that $\dHDY{R}{B} \leq \ell \leq (1+\eps)\dHDY{R}{B}$.
\end{theorem}

\begin{remark}
    \remlab{h_is_new}%
    The result about the Hausdorff distance was implicitly stated in \remref{h_really} in the journal version of this paper, but not explicitly. We described it here in detail for completeness.
\end{remark}

\subsubsection{Exact nearest neighbor distances, and furthest %
   nearest neighbor}
\seclab{f_n_n}

Using \thmref{k_th_m_nn} with $m=1$ and $\eps=1$, results in a $2$-approximation, in $O(n)$ time, to the \emphi{$k$\th nearest neighbor distance} of $\PntSet$.  We now show how this can be converted into an exact solution.

So we have a quantity $r$ that is larger than the $k$\th nearest neighbor distance, but at most twice as big. We build a grid with a cell diameter being $r/4$. Clearly, any point whose nearest neighbor distance is at least the $k$\th nearest neighbor distance must be the only point in its grid cell.  For each such point $\pnt$, we compute its distance to all the points stored in its neighborhood $\GridNbr{\pnt}{r}$, by scanning the list of points associated with these cells.  This computes for these ``lonely'' points their exact nearest neighbor distance.  The $k$\th smallest nearest neighbor distance will then be the $(n-k)$\th largest distance computed.  Clearly, every cell's list of points gets scanned a constant number of times, so overall the running time is linear. We summarize the result.

For a point set $\PntSet$, the \emphi{furthest nearest neighbor distance}, is the maximum distance of a point of $\PntSet$, from the remaining points. Formally, it is $\max_{\pnt \in \PntSet} \ddX{\pnt}{\PntSet}$.

\begin{theorem}
    Let $\PntSet$ be a set of $n$ points in $\Re^d$. For an integer $k$, $1\leq k\leq n$, one can compute exactly, in expected linear time, the $k$\th nearest neighbor distance in $\PntSet$.  In particular, setting $k=n$, one can compute the furthest nearest neighbor distance exactly, in expected linear time.
\end{theorem}

\subsection{The spanning forest partitions, and %
   the \TPDF{$k$}{k}\th %
   longest \TPDF{\MST}{MST} %
   edge}
\seclab{longest_MST}

The net computation provides a natural way to partition the data into clusters and leads to a fast clustering algorithm. One alternative partition scheme is based on distance connectivity.

\begin{definition}
    \deflab{connectivity}%
    For a set of points $\PntSet$ and a number $r > 0$, let $\CCX{\PntSet}{r}$ be the \emphi{$r$-connectivity clustering} of $\PntSet$.

    \begin{figure}[t!]
        \centering \includegraphics{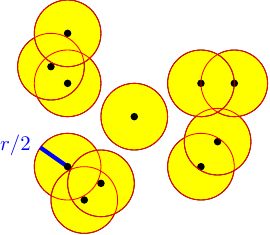}%
        \caption{}
        \figlab{connectivity}
    \end{figure}

    Specifically, it is a partition of $\PntSet$ into connected components of the \MST of $\PntSet$ after all edges strictly longer than $r$ are removed from it (alternatively, these are the connected components of the intersection graph where we replace every point of $\PntSet$ by a disk of radius $r/2$ centered at that point), see \figref{connectivity}.  As such, a set in such a partition of $\PntSet$ is referred to as a \emphi{connected component}.

    Consider two partitions $\Partition, \PartitionA$ of $\PntSet$.  The partition $\Partition$ is a refinement of $\PartitionA$, denoted by $\Partition \sqsubseteq \PartitionA$, if for any set $X \in \Partition$, there exists a set $Y \in \PartitionA$ such that $X \subseteq Y$.
\end{definition}

\begin{lemma}%
    \lemlab{longest_edge_partition}%
    Given a set $\PntSet \subseteq \Re^d$, and parameters $r$ and $\eps$, one can compute, in $O\pth{n/\eps^d}$ time, a partition $\Partition$, such that $\CCX{\PntSet}{r} \sqsubseteq \Partition \sqsubseteq \CCX{\PntSet}{(1+\eps)r}$.
\end{lemma}
\begin{proof}
    Build a grid with every cell having a diameter $\eps r/4$. For every point in $\PntSet$, mark all the cells in distance $r/2$ from it. That takes $O(n/\eps^d)$ time. Next, for every marked cell, create a node $v$, and let $V$ be the resulting set of nodes. Next, create a bipartite graph $\Graph = (V\cup \PntSet, E)$ connecting every node of $V$ to all the points of $\PntSet$, marking its corresponding grid cell.  Now, compute the connected components in $\Graph$, and for each connected component, extract the points of $\PntSet$ that belong to it. Clearly, the resulting partition $\Partition$ of $\PntSet$ is such that any two points in distance $\leq r$ are in the same set (since for such points there must be a grid cell in distance $\leq r/2$ from both).  Similarly, if the distance between two subsets $X, Y \subseteq \PntSet$ is at least $(1+\eps)r$, then they are in different connected components of $\Partition$.  In particular, for $X$ and $Y$ to be in the same component, there must exist points $x\in X$ and $y\in Y$ that mark the same grid cell.  However, such a grid cell would need to be at a distance of at most $r/2$ from both $x$ and $y$, implying $\distX{x}{y}\leq r/2+r/2+\eps r/4 < (1+\eps)r$, a contradiction.  Thus, this is the desired partition.  Clearly, this takes $O(n/\eps^d)$ time overall.
\end{proof}

\begin{theorem}%
    \thmlab{m_s_t_k_edge}%
    Given a set of $n$ points $\PntSet$ in $\Re^d$, and parameters $k$ and $\eps$, one can output, in expected $O\pth{n/\eps^d}$ time, a $(1+\eps)$-approximation to the $k$\th longest (or $k$\th shortest) edge in the \MST of $\PntSet$.
\end{theorem}

\begin{proof}
    Consider the cost function $\fLEMST\pth{\PntSet, k}$ which returns the $k$\th longest edge of the \MST of $\PntSet$.  It is easy to verify that $\fLEMST\pth{\PntSet, k}$ is the minimum $r$ such that $\CCX{\PntSet}{r}$ has $k$ connected components. For the sake of simplicity of exposition, we assume all the pairwise distances of points in $\PntSet$ are distinct.  We need to fill in/verify the properties of \defref{N_D_P}: %
    \smallskip
    \begin{compactitem}[\ItemSpacing]
        \itemPDecider The $(1+\eps$)-decider for $\fLEMST\pth{\PntSet,k}$ uses \lemref{longest_edge_partition}. Specifically, for a specified $r > 0$, compute a partition $\Partition$, such that $\CCX{\PntSet}{r} \sqsubseteq \Partition \sqsubseteq \CCX{\PntSet}{(1+\eps)r}$.  Now, if $\Partition$ has $\leq k$ connected components, then $\fLEMST\pth{\PntSet,k} \leq (1+\eps)r$. Similarly, if $\Partition$ has $> k$ connected components, then $\fLEMST\pth{\PntSet,k} > r$. By calling on two values and using slightly smaller $\eps$, it is straightforward to get the desired $(1+\eps)$-decider, as was done in \lemref{k_dist}.

        \itemPLipschitz %
        Let $\PntSet'$ be some $\Delta$-\translation of $\PntSet$, and let $\Graph$ (resp. $\Graph'$) denote the complete graph on $\PntSet$ (resp. $\PntSet'$). For an edge $e\in E\pth{\Graph}$ (or in $E\pth{\Graph'}$), let $\ell(e)$ denote the length of $e$.  For the sake of simplicity of exposition, we consider $\PntSet'$ and $\PntSet$ to be multi-sets of the same cardinality (i.e., $n = \cardin{\PntSet} = \cardin{\PntSet'} > k$).  (Note this implies some of the edges may have length $0$.)  For a point $\pnt\in \PntSet$, let $\pnt'$ denote its image in $\PntSet'$ after the $\Delta$-\translation.  Let $e_1, \ldots, e_{n-1}$ be the edges in the \MST of $\Graph$ sorted by increasing order by their length.  Let $e_1', \ldots, e_{n-1}'$ denote the corresponding edges in $\Graph'$, i.e. for $e=uv$, the corresponding edge is $e' = u'v'$.

        Observe that the subgraph of $\Graph'$ having $e_1', \ldots, e_{i}'$ as edges, is a forest with $n-i$ connected components.  As such, we have that
        \begin{equation*}
            \fLEMST\pth{\PntSet', n-i} \leq \max_{j=1}^i \ell(e_j')
            \Bigl.
        \end{equation*}
        Indeed, by the correctness of Kruskal's algorithm, the first $i$ edges of the \MST form such a forest with the maximum weight edge being minimum (among all such forests).  As such, we have that
        \begin{align*}
          \fLEMST\pth{\PntSet', n-i}%
          &\leq %
            \max_{j=1}^i \ell(e_j')%
            \leq%
            \max_{j=1}^i \pth{\ell(e_j) + 2\Delta}%
            =%
            2\Delta + \max_{j=1}^i \ell(e_j)%
          \\%
          &=%
            2\Delta + \fLEMST\pth{\PntSet, n-i}.
        \end{align*}
        A symmetric argument implies that
        \begin{equation*}
            \fLEMST\pth{\PntSet, n-i}
            \leq%
            2\Delta +
            \fLEMST\pth{\PntSet', n-i}.
        \end{equation*}
        Setting $i=n-k$, we have that $\cardin{ \fLEMST\pth{\PntSet, k} - \fLEMST\pth{\PntSet', k} } \leq 2\Delta$, as desired.

        \itemPPrune Consider a point $\pnt \in \PntSetA$, such that $\ddX{\pnt}{\PntSetA} \geq r$, and $r > \fLEMST\pth{\PntSetA, k}$. This implies that all the \MST edges adjacent to $\pnt$ are of distance $\geq r$.  Assume there are $j$ such edges, and observe that removing these edges creates $j+1$ connected components of the \MST, such that each connected component is of distance $\geq r$ from each other (indeed, otherwise, one of the removed edges should not be in the \MST, as can be easily verified). Thus, if we delete $\pnt$, and recompute the \MST, we will get $j$ new edges that replace the deleted edges, all of these edges are of length $\geq r$. That is, deleting $\pnt$, and decreasing $k$ by one, ensures that the target value remains the same in the pruned instance.
    \end{compactitem}

    \smallskip

    This establish that $\fLEMST\pth{\PntSet, k}$ is an $(1+\eps)$-\NDP, and by plugging this into \thmref{result_eps}, we get the desired result.
\end{proof}

\begin{remark}
    Note that in the above, as in other parts of the paper, we treat a point of weight $\omega$ as $\omega$ copies of the same point. This naturally leads to edges of length zero in the \MST, which are being handled correctly by the algorithm.%
\end{remark}

\subsection{Computing the minimum cluster}

\subsubsection{Sketchable families}
\seclab{sketchable}

\begin{definition}[Upward Closed Set System]
    Let $\PntSet$ be a finite set of ground elements, and let $\Family$ be a family of subsets of $\PntSet$.  Then $(\PntSet, \Family)$ is an \emphi{upward closed set system} if for any $X \in \Family$ and any $Y \subseteq \PntSet$, such that $X \subseteq Y$, we have that $Y \in \Family$.  Such a set system is a \emphi{sketchable family}, if for any set $S \subseteq \PntSet$ there exists a constant size \emphi{sketch} $\descrip{S}$ such that:
    \begin{compactenum}[\quad(A)]
        \item For any $S, T \subseteq \PntSet$ that are disjoint, $\descrip{S\cup T}$ can be computed from $\descrip{S}$ and $\descrip{T}$ in $O(1)$ time.  We assume the sketch of a singleton can be computed in $O(1)$ time, and as such, the sketch of a set $S\subseteq \PntSet$ can be computed in $O(|S|)$.

        \item There is a membership oracle for the set system based on the sketch. That is, there is a procedure $\oracle{\cdot}$ such that given the sketch of a subset $\descrip{S}$, $\oracle{\descrip{S}}$ returns whether $S\in \Family$ or not, in $O(1)$ time.
    \end{compactenum}
\end{definition}

An example for such a sketchable family is the set system $(\PntSet, \Family)$, where $S \subseteq \PntSet$ is in $\Family$ if $\cardin{S} \geq 10$. Here, the sketch of the set is simply the number of elements in the set, and combining two sketches $\descrip{S}$ and $\descrip{T}$ is adding the numbers to get $\descrip{S \cup T}$ (for $S\cap T =\emptyset$).

We will be interested in two natural problems induced by such a family:
\begin{inparaenum}[(i)]
    \item smallest cluster -- find the smallest set in the family with certain properties, and
    \item min-max clustering -- find disjoint sets in the family such that they cover the original set, and the maximum price of these sets is minimized.
\end{inparaenum}

\begin{remark}
    \remlab{monotonicity_sketch}%
    In the following, we consider sketchable families $(\PntSet, \Family)$, where $\PntSet$ is a set of Euclidean points.  In this case, $\PntSet = \brc{\pnt_1, \dots, \pnt_n}$ can be viewed as a set of elements, where each $\pnt_i$ is additionally given a realization at some location in $\Re^d$.  Specifically, changing the locations of the realizations does not change the subsets in $\Family$.  This is required to ensure the \PLipschitzRef property holds.  Specifically, sketches remain the same after a $\Delta$-\translation.  Furthermore, we do not have to maintain the sketches under deletions (of the underlying points), as pruned points correspond to complete sketches that are being thrown away, see \remref{prune_sketches}.
\end{remark}

\begin{remark}
    Note that the sketches do not need to have $O(1)$ size or the oracle to run in $O(1)$ time. However, assuming otherwise will affect the running time of \ndpAlg if it is used to solve a problem involving a sketchable family.
\end{remark}

\begin{example}
    Consider associating a positive $k$-dimensional vector $\vecA_\pnt$ with each point $\pnt \in \PntSet$ (a vector is \emphi{positive} if it is non-zero, and all its coordinates are non-negative). A \emphi{positive linear inequality} is an inequality of the form $\sum_i \alpha_i x_i \geq c$, where the coefficients $\alpha_1, \ldots, \alpha_k$ are all non-negative. For such a linear inequality, consider the set system $(\PntSet, \Family)$, where a set $\PntSetB$ is in $\Family$ if the linear inequality holds for the vector $\sum_{\pnt \in \PntSetB} \vecA_\pnt$. Clearly, this family is sketchable, the sketch being the sum of the vectors associated with the points of the set (here $k$ is assumed to be a constant).

    It is easy to verify that sketchable families are closed under finite intersection. Specifically, given a collection of $m$ such positive inequalities, the family of sets such that their sketch vector complies with all these inequalities is a sketchable family (of course, checking if a set, given its sketch, is in the family naively would take $O(m k)$ time).

    As a concrete application, consider the scenario where every element in $\PntSet$ has $k=4$ attributes. One might be interested in subsets, $\PntSetB$, such that the sum of the first two attributes of the vector $\sum_{\pnt \in \PntSetB} \vecA_\pnt$ is at least $1$, and the sum of the last two attributes is at least $2$.

    Of course, in general, a membership oracle for the attributes space that has the property that if $\vecA$ is valid, then $\vecA+\vecA'$ is also valid, for any positive $\vecA'$, would define a sketchable family. As a concrete example, consider the non-linear (and not convex) condition that the sum of at least two attributes is larger than $1$. Clearly, this defines a sketchable family.
\end{example}

Clearly, the above definition of sketchable family is very general and widely applicable.  Next, we show how \ndpAlg can be used to approximate certain objectives over sketchable families.

\subsubsection{Min cluster}
\seclab{min_cluster}

We now consider the problem of minimizing the cluster radius of a subset of a given point set subject to inclusion in a sketchable family.  Specifically, we consider the case when the cluster is defined by a ball of radius $r$ or when the cluster is defined by a connected component of $\CCX{\PntSet}{r}$ for a radius $r$. Note that as a ball or component grows, both the set of points the cluster contains, and the inclusion of this set of points in the sketchable $\Family$ are monotone properties.  This correspondence is what allows us to apply our general framework.

\begin{theorem}
    For a set of $n$ points $\PntSet \subseteq \Re^d$, and a sketchable family $(\PntSet, \Family)$, one can $(1+\eps)$-approximate, in expected $O\pth{n/\eps^d}$ time, the radius of the smallest ball, $\ball$, such that $\ball \cap \PntSet \in \Family$.
\end{theorem}
\begin{proof}
    Let $f(\PntSet)$ be the diameter of the smallest ball $\ball$ in $\Re^d$ such that $\PntSet \cap \ball \in \Family$. We claim that $f(\cdot)$ is an \NDP, see \defref{N_D_P}. The \PLipschitzRef and \PPruneRef properties readily hold for $f(\cdot)$.  Specifically, $f(\PntSet)$ was chosen to be the diameter rather than the radius so that the \PPruneRef property holds. Note that below, for the \PDeciderRef property, we refer to approximating the radius, which gives us the same approximation to the diameter.

    For the \PDeciderRef property, given $r$ and $\eps>0$, construct a grid with cells of diameter $\eps r /4$, and register each point of $\PntSet$ in all the grid cells in distance at most $r$ from it.  If there is a ball, $\ball$, of radius $r$ such that $\ball\cap \PntSet \in \Family$ then the set of points registered with the grid cell containing the center of this ball will be a superset of $\ball \cap \PntSet$ and hence the set is in $\Family$, by the upward closed property of sketchable families.  Moreover, the set registered at this grid cell requires a ball of radius at most $r + 2\eps r/ 4$ (centered at any point in this grid cell) to cover it.  Thus, the decision procedure checks the set of points associated with each grid cell to see whether or not it is in $\Family$.  The definition of sketchable families implies this can be done in linear time in the total sizes of these sets (i.e., $O(n/\eps^d)$).  Furthermore, if there is no ball of radius $(1+\eps)r$ whose point set is in $\Family$, then the decision procedure would fail to find such a cell. Thus, this is a $(1+\eps)$-decision procedure, and its running time is $O(n/\eps^d)$.

    Plugging this into the algorithm of \thmref{result_eps} implies the result.
\end{proof}

\begin{remark}
    In the algorithm of the above theorem, as the algorithm progresses, every weighted point that is computed also has an associated sketch from the original points that it corresponds to.  One can view this collection of sketches as either an attribute in addition to the weights or as being passed down via the context.
\end{remark}

The following is a sample of what the above theorem implies.

\begin{corollary}
    We can $(1+\eps)$-approximate, in $O(n /\eps^d)$ time, the following problems for a set of $n$ points in $\Re^d$:
    \begin{compactenum}[(A)]
        \item The smallest ball containing $k$ points of $\PntSet$.

        \item The points of $\PntSet$ are weighted, and there is an additional threshold $\alpha$. Compute the smallest ball containing points with total weight at least $\alpha$.

        \item If the points of $\PntSet$ are colored by $k$ colors, the smallest ball containing points of $\PntSet$, such that they are colored by at least $t$ different colors. (Thus, one can find the smallest non-monochromatic ball [$t=2$], and the smallest ball having all colors [$t=k$].) The running time is $O\pth{n k /\eps^d}$, as the sketch here is a $k$-dimensional vector.
    \end{compactenum}
\end{corollary}

\begin{remark}
    \remlab{prune_sketches}%
    A careful inspection of our algorithm reveals that we only merge sketches or throw them away.  Specifically, a pruning stage of our algorithm corresponds to throwing away some of the sketches, and netting corresponds to merging together some of the sketches.  See also \remref{monotonicity_sketch}.
\end{remark}

\subsubsection{Min cluster in the spanning forest}
\seclab{min_cluster_c}

Note that $\CCX{\PntSet}{r}$ is a monotone partition as $r$ increases (i.e., components merge, but never split apart).  It is natural to ask what the minimum $r$ is, for which there is a connected component in $\CCX{\PntSet}{r}$ that is in a sketchable family.

\begin{theorem}
    \thmlab{min_cluster_MST}%
    For a set of $n$ points $\PntSet \subseteq \Re^d$, and a sketchable family $\,(\PntSet, \Family)$, one can $(1+\eps)$-approximate, in expected $O\pth{n/\eps^d}$ time, the minimum $r$, such that there is a connected component in $\CCX{\PntSet}{r}$ that is in $\Family$.
\end{theorem}

\begin{proof}
    The target function $\target(\PntSet)$ is the smallest $r$ such that $\CCX{\PntSet}{r}$ contains a set that is in $\Family$. Let $\ropt = \target(\PntSet)$.  The algorithm first checks if any of the input points by themselves have the desired property. If so, then $\ropt=0$ and the algorithm returns $0$. (For the case that $\PntSet$ is a multiset, see \remref{multi}.)  As such, we can assume that $f(\PntSet) > 0$, and by \corref{same_1} (C), $\target(\cdot)$ is always non-zero during the algorithm execution.  Again, we need to verify the properties of \defref{N_D_P} for $\target(\cdot)$:
    \begin{compactitem}[\ItemSpacing]
        \itemPDecider %
        One can $(1+\eps/4)$-approximate the connected components of $\CCX{\PntSet}{r/(1+\eps/3)}$, using the algorithm of \lemref{longest_edge_partition}, and for each approximate connected component, use its sketch to decide if any of them is in $\Family$. If so, return that $r$ is larger than the optimal value. Otherwise, perform a $(1+\eps/3)$-approximate computation of $\CCX{\PntSet}{ r}$, and if one of its clusters is in $\Family$, then return that the desired value is in the interval $[r/(1-\eps/3), r(1+\eps/3)]$. Otherwise, return that $r$ is smaller than the optimal value. Clearly, this is $(1+\eps)$-decider that works in $O(n/\eps^d)$ time.%

        \itemPLipschitz %
        Let $r' = f(\PntSet)$.  By the definition of $f(\cdot)$, there exists $k$, such that $r' = \fLEMST\pth{\PntSet, k}$. Now, arguing as in the proof of \thmref{m_s_t_k_edge} implies that $f(\cdot)$ has the desired \PLipschitzRef property.

        \itemPPrune %
        Consider a point $\pnt \in \PntSetA$, such that $\ddX{\pnt}{\PntSetA} \geq r$, and $r > f\pth{\PntSetA} > 0 $.  Now, $\pnt$ by itself can not have the desired property, as this would imply that $\target(\PntSetA) =0$. This implies that in $\CCX{\PntSetA}{\ropt}$ the point $\pnt$ is by itself (and not in any cluster that is in $\Family$). As such, throwing $\pnt$ away does not change the target function. Formally, $f\pth{\PntSetA} = f\pth{\PntSetA \setminus \brc{\pnt}}$.
    \end{compactitem}
    Plugging this into \thmref{result_eps} implies the result.
\end{proof}

One natural application for \thmref{min_cluster_MST} is for ad hoc wireless networks. Here, we have a set $\PntSet$ of $n$ nodes and their locations (say in the plane), and each node can broadcast in a certain radius $r$ (the larger the $r$, the higher the energy required, so naturally we would like to minimize it). It is natural now to ask for the minimum $r$ such that one of the connected components in the resulting ad hoc network has some desired property.  For example, in $O(n/\eps^d)$ time, we can $(1+\eps)$-approximate the smallest $r$ such that:
\begin{compactenum}[\qquad(A)]
    \item One of the connected components of $\CCX{\PntSet}{r}$ contains half the points of $\PntSet$, or more generally, if the points are weighted, then one of the connected components contains points of total weight at least $\alpha$, for a prespecified $\alpha$.

    \item If the points are colored, the desired connected component contains all the colors (for example, each color represents some fraction of the data, and the cluster can recover the data if all the pieces are available), or at least two colors, or more generally, a different requirement on each color.
\end{compactenum}

\subsection{Clustering for monotone properties}
\seclab{cluster_monotone}

\begin{definition}[Min-Max Clustering]
    We are given a sketchable family $\pth[]{\PntSet,\Family}$, and a cost function $g: 2^\PntSet \rightarrow \Re^+$.  We are interested in finding disjoint sets $S_1, \ldots, S_m \in \Family$, such that
    \begin{inparaenum}[(i)]
        \item $\bigcup_i S_i = \PntSet$, and
        \item $\max_i g(S_i)$ is minimized.
    \end{inparaenum}
    We will refer to the partition realizing the minimum as the \emphi{optimal clustering} of $\PntSet$.
\end{definition}

\begin{theorem}
    Let $\PntSet$ be a set of points in $\Re^d$, and let $\pth[]{\PntSet,\Family}$ be a sketchable family.  For a set $\PntSetA \in \Family$, let $\rmin(\PntSetA)$ be the radius of the smallest ball centered at a point of and enclosing $\PntSetA$.  One can $(4+\eps)$-approximate, in expected $O(n/\eps^d)$ time, the min-max clustering under $\rmin$ of $\PntSet$.

    That is, one can cover $\PntSet$ by a set of balls, and assign each point of $\PntSet$ to one of these balls, such that the set of points assigned to each ball is in $\Family$, and the maximum radius of any of these balls is a $(4+\eps)$-approximation to the minimum radius used by any such cover.

    \thmlab{cluster_min_max}
\end{theorem}
\begin{proof}
    Let $\PartitionOpt$ be the optimal partition with radius $\ropt$, and consider an $r$-net $\NetSet$ for $r \geq 4\ropt$, computed using \corref{valid}. Consider a point $\pnt \in \NetSet$, and let $\PntSet_\NetSet[\pnt]$ be the set of points of $\PntSet$ assigned to $\pnt$ by the nearest net-point assignment.

    Next, consider the cluster $\PntSetA \in \PartitionOpt$ that contains it. Clearly, $\diameterX{\PntSetA} \leq 2\ropt$, and the distance of $\pnt$ from all other net points in $\NetSet$ is at least $4\ropt$. It follows that $\PntSetA \subseteq \PntSet_\NetSet[\pnt]$, and since $\PntSetA \in \Family$, it follows that $\PntSet_\NetSet[\pnt] \in \Family$.

    A $4$-decider for this problem works by computing the $4r$-net $\NetSet$, and for each $\pnt \in \NetSet$, checking the sketchable property for the set $\PntSet_\NetSet[\pnt]$. It is easy to verify that the properties of \defref{N_D_P} hold in this case. In particular, throwing a far-away isolated point corresponds to a cluster that already fulfills the monotone property, and it is too far away to be relevant. Namely, computing $\ropt$ is an \NDP and so plugging this into \thmref{result_eps} implies the result.
\end{proof}

\subsubsection{Lower bounded center clustering}
\seclab{cluster_monotone_examples}

If the required sketchable property is that every cluster contains at least $k$ points, then \thmref{cluster_min_max} approximates the \emphi{lower bounded center} problem. That is, one has to cover the points by balls, such that every cluster (i.e., points assigned to a ball) contains at least $k$ points. The price of this clustering is the radius of the largest ball used. A $2$-approximation to this problem is known via the usage of flow \cite{apftk-aac-10}, but the running time is super-quadratic. Recently, the authors showed a similar result to \thmref{cluster_min_max} with running time (roughly) $O(n \log n)$ \cite{ehr-fclb-12}.  This paper also shows that this problem cannot be approximated better than (roughly) $1.8$ even for points in the plane.

\begin{corollary}
    Let $\PntSet$ be a set of points in $\Re^d$, and let $k$ and $\eps>0$ be parameters. One can $(4+\eps)$-approximate the lower bounded center clustering in $O(n/\eps^d)$ time.
\end{corollary}

\subsubsection{Other clustering problems}

One can plug in any sketchable family into \thmref{cluster_min_max}. For example, if the points have $k$ colors, we can ask for the min-max radius clustering, such that every cluster contains (i) all colors, (ii) at least two different colors, or (iii) a different requirement on each color, etc.

As another concrete example, consider that we have $n$ customers in the plane, and each customer is interested in $k$ different services (i.e., there is a $k$-dimensional vector associated with each customer specifying their demand). There are $t$ types of service centers that can be established. Still, each such center type requires a minimum level of demand in each of these $k$ categories (i.e., each type is specified by a minimum demand $k$-dimensional vector, and a set of customers can be the user base for such a service center if the sum of their demand vectors is larger than this specification). The problem is to partition the points into clusters (of minimum maximum radius), such that for every cluster, there is a valid service center assigned to it. Clearly, this falls into the framework of \thmref{cluster_min_max}, and can be $(4+\eps)$-approximated in $O(nkt/\eps^d)$ time.

\subsubsection{Clustering into spanning forests}
\seclab{cluster_spanning}

One can get a similar result to \thmref{cluster_min_max} for the connectivity version of clustering of $\PntSet$.  Formally, a set of points $\PntSetA \subseteq \PntSet$ is \emphi{$r$-valid} if $\PntSetA$ is contained in some set of $\CCX{\PntSet}{r}$. Given a sketchable family $(\PntSet,\Family)$, a partition $\Partition$ of $\PntSet$ is an \emphi{$r$-connected clustering} if all the sets in $\Partition$ are in $\Family$, and are $r$-valid.

\begin{theorem}%
    \thmlab{cluster_connected}%
    Let $\PntSet$ be a set of points in $\Re^d$, and let $\pth[]{\PntSet,\Family}$ be a sketchable family. One can $(1+\eps)$-approximate $\ropt$, in expected $O(n/\eps^d)$ time, where $\ropt$ is the minimum value such that there is a $\ropt$-connected clustering of $\PntSet$.
\end{theorem}
\begin{proof}
    Let $\target(\cdot)$ be the target function in this case. We verify the properties of \defref{N_D_P}: %
    \smallskip%
    \begin{compactitem}[\ItemSpacing]
        \itemPDecider Given $r$, we use \lemref{longest_edge_partition} to compute a partition $\Partition$ of $\PntSet$ such that $\CCX{\PntSet}{r} \sqsubseteq \Partition \subseteq \CCX{\PntSet}{(1+\eps)r}$.  If the sketchable property holds for each cluster of $\Partition$, then return that $r$ is too large.  Otherwise, return that it is too small. As for the quality of this decider, observe that if the optimal partition has a cluster $\PntSetA$ that uses points from two different clusters of $\Partition$, then $\PntSetA$ is not $r$-valid, as otherwise these two points would be in the same cluster of $\Partition$ (namely, $\ropt > r$).

        \itemPLipschitz Follows easily by arguing as in \thmref{m_s_t_k_edge}.

        \itemPPrune If an isolated point exists, then it can be thrown away because the cluster of original points it corresponds to is a valid cluster that can be used in the final clustering. It does not interact with any other clusters.
    \end{compactitem}
    \smallskip%
    Plugging this into \thmref{result_eps} now implies the result.
\end{proof}

A nice application of \thmref{cluster_connected} is for ad hoc networks. Again, we have a set $\PntSet$ of $n$ wireless clients, and some of them are base stations; that is, they are connected to the outside world. We want to find the minimum $r$, such that each connected component of $\CCX{\PntSet}{r}$ contains a base station.

\subsection{Smallest non-zero distance}
\seclab{closest_non_zero_dist}

Given a point set $\PntSet$, it was already shown in \secref{f_n_n} that, with a small amount of post processing, \ndpAlg can be used to compute the closest pair distance exactly in expected linear time.  If one allows $\PntSet$ to contain duplicate points (i.e. $\PntSet$ is a multiset) then a couple modifications to $\ndpAlg$ must be made.  First, modify the algorithm so that for the selected point $\pnt$, it finds the closest point to $\pnt$ that has a non-zero distance from it.  Second, modify \delete (see \lemref{delete}) so that a singleton cell means all points in the cell have the same location, rather than meaning the cell contains a single point. This can easily be checked for all cells in overall linear time.  With these two modifications, the above algorithm works verbatim, and we get the following.

\smallskip

\begin{lemma}%
    \lemlab{non_zero}%
    Let $\PntSet$ be a multiset of weighted points in $\Re^d$.  Then one can solve the \distinct problem {\emph{exactly}} for $\PntSet$, in expected linear time.  In particular, if $\PntSet$ contains no duplicates, then this corresponds to computing the closest pair distance.
\end{lemma}

Interestingly, the algorithm of \lemref{non_zero} is a prune-and-search algorithm, as the net stage never gets executed. Observe that it is not hard to extend the algorithm of Golin \etal{} \cite{grss-sracp-95} to solve this variant, and the result of \lemref{non_zero} is included in the paper only for the sake of completeness. In particular, the resulting (simplified) algorithm is similar to the algorithm of Khuller and Matias \cite{km-srsac-95}.

\section{Linear Time with high probability}
\seclab{high_prob}

In \secref{framework}, we presented an algorithm (\ndpAlg, see \figref{a_algorithm}) that provides a constant factor approximation to a wide variety of problems in expected linear time.  Here we present an extension of that algorithm that runs in linear time with high probability.

Our algorithm works by computing a quantity that is not directly related to the problem at hand (but in linear time!), and then relating this quantity to the desired quantity, enabling us to compute it in linear time. Specifically, the resulting algorithm is somewhat counterintuitive and bizarre. As such, to make the algorithm description more accessible to the reader, we introduce the basic ideas in \secref{h_p_introduction}.  We describe the algorithm in detail in \secref{high_Prob_algorithm}, and analyze it in \secref{high_Prob_analysis}.

\subsection{Preliminaries}
\seclab{h_p_introduction}

\paragraph{Notation.}
See \tabref{notation} for a summary of the notation used in this part of the paper.  As it is sufficient for our purposes, for simplicity (in this section) $\PntSet$ will always be an unweighted points set, i.e.  $n= \WeightX{\PntSet} = \cardin{\PntSet}$.  The algorithm can be easily extended to handle the weighted case.

\begin{table}[t]
    \begin{minipage}{\linewidth}
        \centerline{%
           \begin{tabular}{|c|c|l|}
             \hline
             $\ddX{\pnt}{\PntSet}$
             &
               $\ds\min_{\pntA\in\PntSet\setminus\brc{\pnt}}
               \distX{\pnt}{\pntA}$
             & \DTbl{%
               Distance from $\pnt$ to its nearest neighbor in
               $\PntSet \setminus \brc{\pnt}$.}%
             \\ \hline $\ddiX{i}{\pnt}{\PntSet}$
             &
             & \DTbl{%
               Distance of $\pnt$ to its $i$\th nearest neighbor in
               $\PntSet\setminus\brc{\pnt}$.} %
             \\ \hline%
             $\ddSetX{\XSet}{\PntSet}$
             &
               $\Set{ \ddX{\pntX}{\PntSet }}{ \pntX \in \XSet }$
             &
               \DTbl{%
               Multi-set of distances between points in $\XSet$ and
               nearest neighbors in $\PntSet$.}%
             \\ \hline%
             $\ddiSetX{i}{\XSet}{\PntSet}$
             & $\Set{ \ddiX{i}{\pntX}{\PntSet} }{ \pntX \in \XSet}$
             & \DTbl{%
               Multi-set of distances between points in $\XSet$
               and their $i$\th
               nearest neighbors in $\PntSet$.}%
             \\ \hline
             $\ddRankX{\delta}{\XSet}{\PntSet}$
             &
             &
               \DTbl{Value of rank
               $\floor{\delta \cardin{\XSet}}$ in
               $\ddSetX{\XSet}{\PntSet}$, for some $\delta \in [0,1]$.}%
             \\ \hline
             $\ddiRankX{i}{\delta}{\XSet}{\PntSet}$
             &
             & \DTbl{Value of rank $\floor{\delta \cardin{\XSet}}$
               in $\ddiSetX{i}{\XSet}{\PntSet}$, for some $\delta \in
               [0,1]$.}%
             \\ \hline
           \end{tabular}
        } \captionof{table}{Notation used. If $\XSet = \PntSet$, the second argument is omitted. Thus, $\ddiRankY{i}{\delta}{\PntSet} = \ddiRankX{i}{\delta}{\PntSet}{\PntSet}$.%
        } \tablab{notation}
    \end{minipage}
\end{table}

\subsubsection{Basic idea}

\paragraph{Sample in the middle.}
Observe that if one could sample a nearest neighbor distance that lies in the interval $[\ddRankY{\delta}{\PntSet}, \ddRankY{1-\delta}{\PntSet}]$ for some fixed constant $\delta \in (0,1)$, then we are done.  Indeed, for each iteration of \ndpAlg for which the sampled distance was near the middle, it is guaranteed that the algorithm either terminates or removes a constant fraction of the points. As such, it is sufficient to present an algorithm which returns a value from this range with probability $\geq 1-1/n^c$, for some sufficiently large constant $c$, where $n = \cardin{\PntSet}$.  Such a value near the median is a \emphi{middle} nearest neighbor distance.

For the rest of this section, we therefore focus on the problem of computing such an approximate middle nearest neighbor distance, in linear time, with high probability.  To turn \ndpAlg into an algorithm that runs in linear time with high probability, we replace \lineref{gen_random} and \lineref{rad} in \figref{a_algorithm} with this new subroutine.

\paragraph{Sampling for the median.}
The problem with the random sampling from $\ddSetY{\PntSet}$ done in \ndpAlg is that, though with good probability, the sampled distance lies near the middle, it does not do so with high probability.  This is the same problem one faces when performing the standard basic randomized median selection.  For median selection, one known solution is rather than randomly sample a pivot, to take a sufficiently large random sample and take the median of the random sample as the pivot. This pivot is good with high probability \cite{mr-ra-95}, and it is this result we wish to mimic.

\paragraph{The challenge.}
Given the above discussion, the natural solution to finding a good middle value from $\ddSetY{\PntSet}$ would be to sample multiple values from $\ddSetY{\PntSet}$ and then take the median.  By the Chernoff inequality, if one samples at least a logarithmic number of values from a set of $n$ values ($\ddSetY{\PntSet}$ in our case), then, with high probability, the median of the sample will lie near the median of the entire set.  (See \lemref{high_Prob} for a formal statement and proof of this observation.)  The problem with this approach is that sampling $O(\log n)$ values from $\ddSetY{\PntSet}$ takes $\Theta(n\log n)$ time (at least naively).  Specifically, while the sampling is easy, computing the associated values requires $\Omega(n)$ time per sampled point (note this is not a problem for median selection, as there we are sampling values directly rather than points from which values are determined).

Therefore, we not only take a sample from which to select our median value, but also a sample from which we will determine what these values are.  Namely, for each point in a $\Theta(\log n)$ sized random sample, compute its nearest neighbor in a second $\Theta(n/\log n)$ sized sample, and then take the median of the computed values, and let this quantity be $\nRad$.  Though this approach takes $\Theta(n)$ time, we need to relate somehow $\nRad$ to the desired middle nearest neighbor distance.

\paragraph{Intuition.}
First note that for most points in the $\Theta(\log n)$ sample, their nearest neighbor in $\PntSet$ will not make it into the $\Theta(n/\log n)$ sample from which we compute nearest neighbor distances.  This implies $\nRad$ is an overestimate of middle nearest neighbor distances in $\PntSet$.  The high-level intuition is that how well $\nRad$ approximates middle distances tells us something about the distribution of $\PntSet$.  There are two cases.

In the first case, most of the points of $\PntSet$ are in small clusters\footnote{Think of constant-size clusters here, although in the algorithm they may be up to $\log^2 n$ sized.  Also, note the clusters are non-trivial, as otherwise we are really in the second case. }.  This case can be detected since here $\nRad$ will be way off, since for any point in the $\Theta(\log n)$ sample, a $\Theta(n/\log n)$ sample is not likely to have a point from the same small cluster.  The advantage is that if points are in small clusters, then one only needs to look within a small cluster to compute a nearest neighbor distance (and this is done using connectivity clustering).%

In the second case, rather than being clustered, $\PntSet$ is roughly evenly distributed.  Hence, intuitively, for any $k$, $\ddiX{k}{\pnt}{\PntSet} \leq k\ddX{\pnt}{\PntSet}$.  Note that for any point in the $\Theta(\log n)$ sample, with high probability we are expected to see one of its $\log^2 n$ nearest neighbors in a $\Theta(n/\log n)$ sample, and so $\nRad \leq \ddiRankY{\log^2 n}{1/2}{\PntSet}$.  Combining these observations, roughly speaking $\nRad \leq (\log^2 n) \ddRankY{1/2}{\PntSet}$.  In other words, we know the median lies in a $\log^2 n$ spread interval, $[\nRad/\log^2 n, \nRad]$.  This case is then resolved using nets and grids (the details here are non-trivial).%

\begin{figure}[t]%
    \phantomsection%
    \setcounter{gx@linecounter}{0}%
    \setcounter{g@linecounter}{0}%
    \AlgorithmAnchor{\si{midNN}}%
    \centerline{%
       \fbox{%
          \begin{nprogram}
              \medianNN{}$(\PntSet)$:\+\\
              $\nRad \leftarrow \estLogDist(\PntSet)$. \\[0.2cm]%
              $\res \leftarrow \deciderRank\pth{ \PntSet,\,
                 \nRad, \,\mathbf{ 3/4} }$\\
              \If $\res = ``\ddRankY{3/4}{\PntSet} \in [x,y]$'' \Then \Return $x$ \>\tabright{\CodeComment{// $y/x = O(1)$}}\\%
              \If $\res = ``\nRad < \ddRankY{3/4}{\PntSet}$'' \Then%
              \Return $\nRad$%
              \linelab{easy_Out}%
              \>\tabright{\CodeComment{//\remref{easy_Out}}} \\[0.2cm]
              \CodeComment{// Must be that $\res = ``\nRad > \ddRankY{3/4}{\PntSet} $''.
              }\\
              $\nxRad \leftarrow \nRad \, /\, 64 \ceil{\ln^2 n} $\\[0.2cm]
              \begin{math}
                  \reslog \leftarrow \deciderRank\pth{\bigl.  \PntSet, \, \nxRad, \, \mathbf{1/2} }
              \end{math}%
              \\
              \If $\reslog = ``\ddRankY{1/2}{\PntSet} \in [x,y]$'' \Then%
              \Return $x$ \>\tabright{\CodeComment{// $y/x =
                    O(1)$}}\\
              \If {$\reslog = ``\nxRad < \ddRankY{1/2}{\PntSet}$''} \Then \\%
              \>\>%
              \linelab{log_Spread}%
              \Return
              \begin{math}
                  \lowSpread\pth{ \bigl. \PntSet, \; \nxRad / 4 , \; 4 \nRad\, }
              \end{math}
              \>\tabright{\CodeComment{//\secref{logarithmicAlg}}}\\[0.3cm]
              \CodeComment{// Must be that $\res_{\log} = ``\nxRad >
                 \ddRankY{1/2}{\PntSet}$''} \\
              \Return \linelab{small} $\smallComp(\nRad, \PntSet)$ \>\tabright{\CodeComment{//\secref{smallAlg}}}%
          \end{nprogram}
       }%
    }%
    \caption{Algorithm for computing a constant factor approximation to a value in the interval $[ \ddRankY{\delta}{\PntSet}, \ddRankY{1- \delta}{\PntSet}]$, for some fixed constant $\delta \in (0,1)$.%
    }
    \figlab{high_prob}%
\end{figure}%

\subsection{The Algorithm}
\seclab{high_Prob_algorithm}

We now present the algorithm in full detail.  Specifically, let \medianNN denote the new algorithm for computing a value in $\Interval = [\ddRankY{\delta}{\PntSet}, \allowbreak \ddRankY{1- \delta }{\PntSet}]$, for some fixed constant $\delta \in (0,1)$.  In \secref{f_n_n} it was shown how to convert a constant factor approximation to such a value into an exact value.  Hence, computing a constant factor approximation to a value in $\Interval$ in linear time, with high probability, suffices for our purposes.

The full algorithm is shown in \figref{high_prob}.  Before analyzing the algorithm, we first define several subroutines that it employs.

\AlgorithmAnchor{decider_rank}%
\subsubsection{\TPDF{\deciderRank}{decider{}Rank}: Counting distances}%

Given a point set $\PntSet$ and a distance $r$, the decision procedure $\deciderRank(\PntSet, r, \alpha)$ decides if the $t$\th nearest neighbor distance in $\ddSetY{\PntSet}$ is smaller than $r$, larger than $r$, or roughly equal to it, where $t=\floor{\alpha \cardin{\PntSet}}$. This decider is implicitly described in \thmref{k_th_m_nn}, but a simpler decider suffices, which we describe here in detail for the sake of clarity. To this end, throw the points of $\PntSet$ into a grid with cell diameter $r$. Now, for every point $\pnt \in \PntSet$, let $w_r(\pnt)$ be the number of the points in $\GridNbr{\pnt}{r}$, not including $\pnt$ itself.  Clearly, $w_r(\pnt)$ can be computed in constant time (given we have first recorded the number of points in each non-empty grid cell). Now, we compute the count
\begin{align*}
  T%
  =%
  \sum_{\pnt \in \PntSet} \frac{w_r(\pnt)}{2},
\end{align*}
in $O( \cardin{\PntSet})$ time. Now, if $T < t$, the procedure returns that $r$ is smaller than the $t$\th nearest-neighbor distance.  Similarly, if $T \geq t$, then the $t$\th nearest-neighbor distance must be at most $3r$. In this case, by also computing the count for the distance $r/3$, one will be able to either return a constant spread interval that contains the desired distance, or that this distance is strictly smaller than $r$. Thus yielding the required decider.

\subsubsection{\TPDF{\estLogDist}{estLogDist}: Sampling a distance}
\seclab{estimate}

Given a set $\PntSet$ of $n$ points in $\Re^d$, take two random samples from $\PntSet$, a sample $\XSet$ of size $\Theta(\log n)$ and a sample $\SSet$ of size $\Theta(n/\log n)$, and compute $\ddRankX{1/4}{\XSet}{\SSet}$.  This can be done in $O(n)$ time, by scanning all of $\SSet$ for each point in $\XSet$ to compute the set of distances $\ddSetX{\XSet}{\SSet}$.  Next, compute the value of rank $\cardin{\XSet}/4$ in this set.  We denote this algorithm by $\estLogDist(\PntSet)$, and it is shown in \figref{med}.

\begin{figure}[t]%
    \AlgorithmAnchor{est_log_dist}%
    \centerline{%
       \fbox{%
          \begin{myprogram}%
              \estLogDist{}$(\PntSet)$:\+\\
              $n=\cardin{\PntSet}$\\
              $\XSet \leftarrow$ random sample of $\PntSet$ of size
              $\ceil{ \cE \ln n }$.\\
              $\SSet \leftarrow$ random sample of $\PntSet$ of
              size $\ceil{\cF n/\ln n}$\\
              Compute $\ddSetX{\XSet}{\SSet}$\\
              \Return the element of rank $\cardin{\XSet} / 4$ in $\ddSetX{\XSet}{\SSet}$.
          \end{myprogram}%
       }%
    }%
    \caption{Roughly estimating $\ddiRankY{O(\log n)}{1/4}{\PntSet} $. Here $\cE$ and $\cF$ are sufficiently large constants.}
    \figlab{med}%
\end{figure}%

\subsubsection{\TPDF{\lowSpread}{lowSpread}: Logarithmic Spread}
\seclab{logarithmicAlg}

The subroutine \lowSpread used by \medianNN (i.e. \lineref{log_Spread} in \figref{high_prob}) handles the case when the desired middle nearest neighbor distance lies in a poly-logarithmic spread interval. The analysis of \lowSpread can be found in \secref{logarithmic}.  In the following, $\PntSet$ is a set of unweighted points such that $[\ddRankY{1/2}{\PntSet}, \allowbreak \ddRankY{3/4}{\PntSet}]\subseteq [r, R]$, for some values $ 0<r\leq R $, where $n=\cardin{\PntSet}$. (Note, in the following, when sampling from a set $\SSet$, we treat it as a set of $\WeightX{\SSet}$ points.)

\begin{algorithm}[\unboldmath{$\lowSpread(\PntSet, r, R )$}]%
    \alglab{spread}%
    \AlgorithmAnchor{log_n_n}%
    ~
    \begin{compactenum}[\qquad (A)]
        \item Compute $\SSet = \netX{r/8}{\PntSet}$.
        \item Throw the points of $\SSet$ into a grid of sidelength $R$.

        \item $\RSample \leftarrow$ random sample from $\SSet$ of $\Theta(\log(n))$ points.

        \item Let $\ZSet$ be an approximation to the set $\ddSetX{\RSample}{\SSet} $.

        Specifically, for each point $\pnt \in \RSample$ look within $\GridNbr{\pnt}{2R}$ for its nearest neighbor, see \defref{grid_stuff}.  If such a neighbor is found, record this distance (which might be zero if $\pnt$ is of weight $>1$), and if no neighbor is found, set this value to $\infty$.

        \item Let $\pnt \in \PntSet$ be the point that corresponds to the number of rank $(5/8)\cardin{\ZSet}$ in $\ZSet$, and let $\tau$ be this number.

        \item Return $\tau$ ($\tau$ is a constant approximation to $\ddX{\pnt}{\PntSet}$).
    \end{compactenum}
\end{algorithm}

\subsubsection{\TPDF{\smallComp}{smallComp}: Many small components}
\seclab{smallAlg}

The subroutine \smallComp, used by \medianNN (i.e. \lineref{small} in \figref{high_prob}), handles the case when the desired middle nearest neighbor distance is considerably smaller than the value sampled by \estLogDist. This corresponds, intuitively, to the situation where the input is mainly clustered in tiny clusters.  The analysis of \smallComp is delegated to \secref{small}.

\begin{algorithm}[\unboldmath{$\smallComp( \nRad, \PntSet)$}]%
    \AlgorithmAnchor{small_comp}%
    \alglab{small}%
    The algorithm works as follows:
    \begin{compactenum}[\quad (A)]
        \item Compute a partition $\ccx$, such that $ \CCX{\PntSet}{\rho} \sqsubseteq \ccx \sqsubseteq \CCX{\PntSet}{2 \rho }$ in linear time, using the algorithm of \lemref{longest_edge_partition}, where $\rho = \nRad / 8m$, where $m = \ceil{\ln^2 n}$ and $n = \cardin{\PntSet}$.

        \item Identify the components of $\ccx$ which are neither singletons nor contain more than $\cC \ln^2 n $ points, where $\cC$ is a sufficiently large constant.

        \item Collect all the points in these components into a set of points $\XSet$ (this set is not empty).

        \item Take a random sample $\RSample$ of $\cD \log n$ points from $\XSet$, where $\cD$ is a sufficiently large constant.

        \item For each point in $\RSample$, compute its nearest neighbor in $\PntSet$ by scanning the non-trivial connected component in $\ccx$ that contains it. Let $\ZSet$ be this set of numbers.

        \item Return the median value in $\ZSet$.
    \end{compactenum}
\end{algorithm}

\subsection{Algorithm Analysis}
\seclab{high_Prob_analysis}

\subsubsection{Why the median of a sample is a good estimate for the median}

The following lemma is by now standard, and we include its proof for the sake of completeness.

\begin{lemma}
    \lemlab{high_Prob}%
    Let $\XSet$ be a set of $n$ real numbers, constants $\delta \in (0,1/2)$, $\alpha \in (0,1)$, and a parameter $t > 0$.  Let $\RSample$ be a random sample of size $t$ picked (with repetition and uniformly) from $\XSet$.  Then, the element of rank $\alpha t$ in $\RSample$ has rank in the interval $\pbrc[]{(1-\delta)\alpha n,\, (1+\delta)\alpha n}$ in the original set $\XSet$, with probability $\geq 1-2\exp(-\delta^2 \alpha t/ 8)$.  In particular, for $\ds t\geq \pth{8c\ln n}/\pth{\delta^2 \alpha } $, this holds with probability $\geq 1-2/n^c$.
\end{lemma}
\begin{proof}
    Let $Y_L$ be the number of values in $\RSample$ of rank $\leq (1-\delta)\alpha n$ in $\XSet$.  Similarly, let $Y_R$ be the number of values in $\RSample$ of rank $\leq (1+\delta)\alpha n$ in $\XSet$.  Clearly
    \begin{math}
        \mu_L = \Ex{Y_L \bigr. } = (1-\delta)\alpha t
    \end{math}
    and $\mu_R = \Ex{ Y_R \bigr.} = (1+\delta)\alpha t\Bigr.$.

    Observe that if $Y_L \leq \alpha t \leq Y_R$ then the element of rank $\alpha t$ in $\RSample$ has rank in $\XSet$ in the interval $\pbrc[]{(1-\delta)\alpha n, (1+\delta)\alpha n}$, as desired.  By the Chernoff inequality \cite{mr-ra-95}, we have that
    \begin{align*}
      \Prob{ \Bigl.Y_L > \alpha t }%
      &=%
        \Prob{Y_L > \frac{ \alpha t}{\mu_L} \mu_L}%
        = %
        \Prob{Y_L > \frac{ \alpha t}{(1 - \delta)\alpha t} \mu_L}%
      \\&
      \leq%
      \Prob{\Bigl. Y_L > (1+\delta) \mu_L}%
      \leq%
      \exp \pth{ - \frac{ \delta^2}{4} \mu_L }%
        \leq%
        \exp \pth{ - \frac{ \delta^2 }{8} \alpha t },
    \end{align*}
    as $1/(1-\delta) > 1+\delta$.  Arguing as above, the other bad
    event probability is
    \begin{align*}
        \Prob{ \Bigl. Y_R < \alpha t }%
        &\leq %
        \Prob{ \Bigl. Y_R < \frac{ \alpha t }{\mu_R} \mu_R }%
        =%
        \Prob{ \Bigl. Y_R < \frac{ \alpha t }{(1+\delta)\alpha t}
           \mu_R }%
        \\&%
        \leq%
        \Prob{ \Bigl. Y_R < \pth{1-\frac{\delta}{2}} \mu_R }%
        \leq%
        \exp\pth{ - \frac{\delta^2}{8}\mu_R} \leq%
        \exp\pth{ - \frac{\delta^2}{8}\alpha t}.%
        \qedhere
    \end{align*}
\end{proof}

\subsubsection{\TPDF{\estLogDist}{estLogDist}: Analysis}
\seclab{est_log_dist_analysis}

\noindent%
\begin{lemma}%
    \lemlab{est_log_dist}%
    Let $\PntSet$ be a set of $n$ points in $\Re^d$. Let $\nRad$ be
    the distance returned by executing $\estLogDist(\PntSet)$, see
    \figref{med}.  With high probability, we have the following:
    \begin{compactenum}[\quad (A)]
        \item For all points $\pnt \in \PntSet$, the set
        $\cardin{ \PntSet \cap \ballCR{\pnt}{r_\pnt} }$ has at most
        $m = \ceil{\ln^2 n}$ points, where
        \begin{inparaenum}[(i)]
            \item $r_\pnt = \ddX{\pnt}{\SSet}$,
            \item $\SSet$ is the random sample used by \estLogDist,
            and
            \item $\ballCR{\pnt}{r_\pnt}$ is the ball of radius
            $r_\pnt$ centered at $\pnt$.
        \end{inparaenum}

        \item Let $\Partition$ be a clustering of $\PntSet$ such that $\CCX{\PntSet}{\rho} \sqsubseteq \Partition \sqsubseteq \CCX{\PntSet}{2\rho}$, as computed by the algorithm of \lemref{longest_edge_partition}, for $\rho = \nRad / 8m$.  Then, at least $(5/8)n$ of the points of $\PntSet$ are in clusters of $\Partition$ that contain at most $m$ points.

        \item $\nRad \geq \ddRankY{1/8}{\PntSet}$.
    \end{compactenum}
\end{lemma}
\begin{figure}[t]
    \centering
    \includegraphics{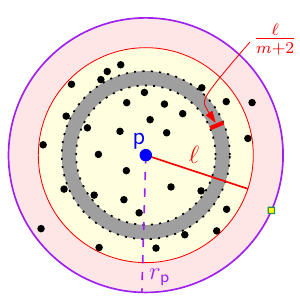}%

    \captionof{figure}{Illustration of case (B) in the proof of
       \lemref{est_log_dist}.}%
    \figlab{proof_b}%
\end{figure}

\begin{proof}
    (A) Fix a point $\pnt \in \PntSet$, and consider a ball $\ball$ centered at $\pnt$ whose interior contains exactly $\ceil{ \ln^2 n}$ points of $\PntSet$. The probability that a random point from $\PntSet$ hits $\ball$ is $\alpha = \ln^2 n / n$. As such, the probability that all the points of $\SSet$ miss it is at most $(1-\alpha)^{\cE n /\ln n} \leq e^{-\cE \ln n} \leq 1/n^{\cE}$, implying the claim, if $\cE$ is sufficiently large.

    (B) By \lemref{high_Prob}, for $\cF$ sufficiently large, we have that $\nRad$ is in the range %
    $[ \ddRankX{1/8} {\PntSet} {\SSet}, \ddRankX{3/8}{\PntSet} {\SSet}]$, %
    and this holds with high probability. That is, for at least $(5/8)n$ points of $\PntSet$ have $\ddX{\pnt}{\SSet} \geq \nRad$, and let $\PntSetB$ be this set of points. For every point $\pnt \in \PntSetB$, we have that the ball $\ball = \ballCR{\pnt}{\nRad} \subseteq \ballCR{\pnt}{r_\pnt}$ contains at most $m$ points of $\PntSet$, by (A). Partitioning $\ball$ into equal radius rings of width $\nRad/(m+2)$, one of them must be empty of points of $\PntSet$. In particular, no cluster of $\CCX{\PntSet}{2 \rho}$ can contain points on both sides of this ring, see \figref{proof_b}. As such, at least $(5/8)n$ of the points of $\PntSet$ are in clusters of $\Partition$ that contain at most $m$ points.

    (C) Follows immediately from (B), as
    $\nRad \geq \ddRankX{1/8}{\PntSet}{\SSet} \geq
    \ddRankY{1/8}{\PntSet}$.
\end{proof}

\begin{remark}
    \remlab{easy_Out}%
    \lineref{easy_Out} of $\medianNN(\PntSet)$ returns $\nRad$ as the approximate middle nearest neighbor distance if $\deciderRank$ returned $\nRad<\ddRankY{3/4}{\PntSet}$.  \lemref{est_log_dist} (C) then implies that
    \begin{math}
        \ddRankY{1/8}{\PntSet}\leq \nRad<\ddRankY{3/4}{\PntSet},
    \end{math}
    with high probability, and so $\nRad$ is a valid approximation.  In the unlikely event that $\ddRankY{1/8}{\PntSet}> \nRad$, one simply runs $\medianNN(\PntSet)$ again.  This event can be determined with the call $\deciderRank(\PntSet, \nRad,1/8)$.
\end{remark}

\subsubsection{\TPDF{\lowSpread}{lowSpread}: Analysis}
\seclab{logarithmic}

During an execution of $\medianNN(\PntSet)$, $\lowSpread(\PntSet, r, R)$ gets called (see \algref{spread}) only if the interval $[r, R]$ contains $\ddRankY{1/2}{\PntSet}, \ldots, \ddRankY{3/4}{\PntSet}$.

\begin{lemma}
    Let $\PntSet$ be a point set in $\Re^d$ of size $n$. Given $r, R \in\Re$, such that (i) $R/r = O\pth{\log^2 n }$, and (ii) $\ddRankY{1/2}{\PntSet}, \ldots, \allowbreak \ddRankY{3/4}{\PntSet} \in [r, R]$, then $\lowSpread(\PntSet, r, R)$ runs in $O(n)$ time, and returns a value which is a $O(1)$-approximation to some nearest neighbor distance $\ddX{\pnt}{\PntSet} \in [\ddRankY{1/2}{\PntSet},\ddRankY{3/4}{\PntSet}]$, with high probability, for some point $\pnt \in \PntSet$.
\end{lemma}

\begin{proof}
    Set $\Delta = r/8$.  Using the notation from \lowSpread, see \algref{spread}, let $\SSet = \netX{\Delta}{\PntSet}$ and $\RSample$ the $\Theta(\log(n))$ size sample from $\SSet$.  By \lemref{high_Prob}, with high probability
    \begin{math}
        \ddRankX{5/8}{\RSample}{\SSet} \in
        \pbrc[]{\ddRankY{1/2}{\SSet}, \ddRankY{3/4}{\SSet}} ,
    \end{math}
    and assume that this is the case.  Since $\SSet$ is an $\Delta$-net of $\PntSet$, by \clmref{change}, this changes distances by at most $2\Delta$, which implies that
    \begin{align*}
        \tau &= \ddRankX{5/8}{\RSample}{\SSet} %
      \in%
        \pbrc[]{\ddRankY{1/2}{\PntSet}-2\Delta, \,
           \ddRankY{3/4}{\PntSet} + 2\Delta }%
            \\&%
            \subseteq %
        \pbrc[]{\frac{1}{2}\ddRankY{1/2}{\PntSet}, \, 2
           \ddRankY{3/4}{\PntSet} }%
        \subseteq \pbrc[]{\frac{r}{2}, 2R}.
    \end{align*}
    by assumption.

    In particular, this implies that all the relevant distances from the computation of this middle value are in the range $[r/2,2R]$. As such, when computing the $(5/8)$-middle element, it is allowable to not compute precisely values outside this interval, as done by the algorithm. This implies that \lowSpread indeed returns the value $\tau$, which is the desired approximation.

    As for the running time, observe that $\SSet$ is an $\Delta$-net, and an open ball with this radius around any point must be empty of any other net point.  In particular, for any point $\pnt \in \SSet$, the maximum number of other points that can be in $\GridNbr{\pnt}{2R}$ is $O\pth{\pth{R/r}^{d} } = O\pth{\log^{2d} n}$.  Therefore, for a point $\pnt \in \RSample$, computing $\ddX{\pnt}{\SSet}$ takes polylogarithmic time (using linear time preprocessing and a grid).  Thus, overall, computing the set of distances $\ddSetX{\RSample}{\SSet}$ takes polylogarithmic time, since $\cardin{\RSample} = \Theta(\log n)$.  Hence, the overall running time is $O(n + \polylog n) = O(n)$.
\end{proof}

\subsubsection{\TPDF{\smallComp}{smallComp}: Small Components}
\seclab{small}

\begin{lemma}
    The call to $\smallComp(\nRad, \PntSet)$ by \medianNN computes a point $\pnt \in \PntSet$, such that $\ddX{\pnt}{\PntSet} \in \pbrc[]{ \ddRankY{1/32}{\PntSet}, \ddRankY{31/32}{\PntSet}}$.  The procedure succeeds and its running time is $O(n)$, with high probability, where $n = \cardin{\PntSet}$.
\end{lemma}
\begin{proof}
    Computing the components of $\ccx$ takes linear time, and given $\ccx$, computing $\XSet$ takes linear time.  For each point in the sample $\RSample$, it takes $O\pth{ \log^2 n }$ time to compute its nearest neighbor since one only has to scan the point's $O(\log^2 n)$ sized connected component in $\ccx$.  Since $\cardin{\RSample} = O(\log n)$, computing the nearest neighbor distances and taking the median takes polylogarithmic time overall.

    Now, $\smallComp(\nRad, \PntSet)$ is called by \medianNN only if the following conditions holds:
    \begin{compactenum}[\qquad (C1)]
        \item $\nRad/64 \ceil{ \ln^2 n} > \ddRankY{1/2} {\PntSet}$.  Namely, at least half of the points of $\PntSet$ have their nearest neighbor in their own cluster in $\CCX{\PntSet}{\rho}$, where $\rho = \nRad / \pth[]{8\ceil{\ln^2 n}}$.  Let $\ZSet_1$ be the set of these points.

        \item Let $\Partition$ be the partition of $\PntSet$ computed by \smallComp. We have that $\CCX{\PntSet}{\rho} \sqsubseteq \Partition \sqsubseteq \CCX{\PntSet}{2\rho}$.  By \lemref{est_log_dist} (B), at least $(5/8) n$ of the points of $\PntSet$ are in clusters of $\Partition$ that contain at most $m$ points (note that some of these points might be a cluster containing only a single point).  Let $\ZSet_2$ be the set of these points.
    \end{compactenum}
    \smallskip%
    As such, $\cardin{\ZSet_1 \cap \ZSet_2} \geq n/8$. Namely, at least $n/8$ of the points of $\PntSet$, are in small clusters and these clusters have strictly more than one point.  Furthermore, $\ZSet_1 \cap \ZSet_2 \subseteq \XSet$, and we conclude that $\cardin{\XSet} \geq n/8$. Now, by \lemref{high_Prob}, with high probability, the returned value is in the range $\ddRankY{1/4}{\XSet}, \ldots, \ddRankY{3/4}{\XSet}$, which implies that the returned value is in the range $\ddRankY{1/32}{\PntSet}, \ldots, \ddRankY{31/32}{\PntSet}$, as desired.
\end{proof}

\subsection{The result}

The analysis in the previous section implies the following.

\begin{lemma}
    \lemlab{high_Prob_Summary}%
    Given a set $\PntSet \subseteq \Re^d$ of $n$ points, $\medianNN(\PntSet)$ computes, a value $x$, such that there is a value of $\ddRankY{1/32}{\PntSet}, \ldots, \ddRankY{31/32}{\PntSet}$ in the interval $[x/c, cx]$, where $c>1$ is a constant. The running time bound holds with high probability.
\end{lemma}

As shown in \secref{f_n_n}, given a constant factor approximation to some nearest neighbor distance, we can compute the corresponding nearest neighbor distance exactly in linear time using a grid computation.  Let \medianNNExact denote the algorithm that performs this additional linear time step after running \medianNN.  As argued above, replacing \lineref{gen_random} and \lineref{rad} in \figref{a_algorithm} with ``$\nRad_i = \medianNNExact( \PntSet_{i-1} )$'', turns \ndpAlg into an algorithm that runs in linear time with high probability.

We thus have the following analogue of \thmref{result_eps}.

\begin{theorem}%
    For a constant $\cDecider > 3/2$, given an instance of a $\cDecider$-\NDP defined by a set of $n$ points in $\Re^d$, with a $\cDecider$-decider with (deterministic or high-probability) running time $O(n)$, one can get a $\cDecider$-approximation to its optimal value, in $O\pth{ n }$ time, with high probability.

    Similarly, for $\eps \in(0,1/2)$, given an instance of a $\,(1+\eps)$-\NDP defined by a set of $n$ points in $\Re^d$, with a $(1+\eps)$-decider having (deterministic, or high probability) running time $O(n/\eps^c)$, then one can $(1+\eps)$-approximate the optimal value for this \NDP, in $O(n/\eps^c)$ time, with high probability.
\end{theorem}

\begin{corollary}
    All the applications of \secref{applications} can be modified so
    that their running time bound holds not only in expectation, but
    also with high probability.
\end{corollary}

\begin{remark}
    Note that the above high probability guarantee is in the size of the input. As such, in later iterations of the algorithm, when the input size of the subproblem is no longer polynomial in $n$ (say, the subproblem size is $O(\polylog n)$), the high probability guarantee is no longer meaningful.

    Fortunately, this is a minor technicality -- once the subproblem size drops below $n / \log n$, the remaining running time of \ndpAlg is $O(n)$, with high probability. A simple direct argument implies this -- by the Chernoff inequality, after $\Theta( \log n)$ iterations, with high probability at least $\log_{16/15} n$ iterations are successful, see the proof of \lemref{time} for details. After this number of successful iterations, the algorithm terminates. This implies that the remaining running time is $O(n)$, with high probability, once the subproblem size is $O(n / \log n)$.

    As such, our new high-probability algorithm is needed only in the initial ``heavy'' iterations of the \ndpAlg algorithm, until the subproblem size drops below $O(n / \log n)$.
\end{remark}

\section{Conclusions}
\seclab{conclusions}

There is further research to be done in investigating the technique presented here.  For example, looking into the implementation of this new algorithm in both standard settings and distributed settings (e.g., MapReduce).  Additionally, since one can now do approximate distance selection in linear time, maybe one can get a speed-up for other algorithms that do (not necessarily point-based) distance selection.

Our framework provides a new way of looking at distance-based optimization problems, in particular through the lens of nets.  We know how to compute nets efficiently for doubling metrics, and it seems one can compute approximate nets in near-linear time for planar graphs.  For example, it seems the new technique implies that approximate $k$-center clustering in planar graphs can be done in near-linear time.  This provides fertile ground for future research.

\paragraph{Open question: When is linear time not possible.}

Another interesting open problem arising out of the high probability linear time algorithm is the following: Given sets $\PntSet$ and $\XSet$ in $\Re^d$, of size $n$ and $O( \log n)$, respectively, how long does it take to approximate the values in the set $\ddSetX{\XSet}{\PntSet}$ (i.e., for each point of $\XSet$ we approximate its nearest-neighbor in $\PntSet$). This problem can be easily solved in $O(n \log n)$ time. The open question is whether it is possible to solve this problem in linear time.

The algorithm in this paper implies that we can compute exactly a value (of a specified rank) in this set in $O(n)$ time (even if $\cardin{\XSet}$ is of size $n$).  We currently have an algorithm with running time $O( n \log \cardin{\XSet}) = O( n \log \log n)$ for approximating $\ddSetX{\XSet}{\PntSet}$.

More generally, this is a proxy to the meta-question of when linear-time algorithms do not exist for distance problems induced by points.

\subsection*{Acknowledgments}

The authors thank Sergio Cabello for valuable discussions. In particular, the deterministic algorithm of \secref{bounded_spread} was suggested by him. The authors also thank the anonymous referees for their insightful and detailed comments.

\begin{sloppypar}
\printbibliography
\end{sloppypar}

\end{document}